\documentclass[final]{siamltex}
\usepackage{url}
\usepackage[normalem]{ulem}
\usepackage{xspace}
\usepackage{amsmath}
\usepackage{amsfonts}
\usepackage{amssymb,amsxtra}
\usepackage{graphicx}
\usepackage{color}
\usepackage{boxedminipage}
\usepackage{algorithmicx,algpseudocode}
\usepackage{algorithm}
\usepackage{verbatim}
\newcommand{\TrAi}{\mathbf{Tr}(A^{-1})}
\newcommand{\Tr}{\mathbf{Tr}}
\newcommand{\onesN}{\mathbf{1}_N}
\newcommand{\colorred}[1]{\textcolor{black}{#1}}
\newcommand{\cred}{\color{black}}
\newcommand{\cblack}{\color{black}}

\title{Hierarchical probing for estimating the trace of the matrix inverse
on toroidal lattices
}
\author{Andreas Stathopoulos \footnotemark[2] 
\and Jesse Laeuchli \footnotemark[2]
\and Kostas Orginos \footnotemark[3]\ \footnotemark[4] }
\begin{document}
\maketitle
\renewcommand{\thefootnote}{\fnsymbol{footnote}}
\footnotetext[2]{Department of Computer Science, College of William and Mary, Williamsburg, Virginia 23187-8795, 
U.S.A.(amrehim@cs.wm.edu)} 
\footnotetext[3]{Department of Physics, College of William and Mary, Williamsburg, Virginia 23187-8795, 
U.S.A.(andreas@cs.wm.edu)}               
\footnotetext[4]{Jefferson National Laboratory, 12000 Jefferson Avenue, Newport News, Virginia, 23606, 
U.S.A.(kostas@jlab.org)}  
\renewcommand{\thefootnote}{\arabic{footnote}}
\begin{abstract}
The standard approach for computing the trace of the inverse of a very large, 
  sparse matrix $A$ is to view the trace as the mean value of 
  matrix quadratures, and use the Monte Carlo algorithm to estimate it.
This approach is heavily used in our motivating application of Lattice QCD.
Often, the elements of $A^{-1}$ display certain decay properties 
  away from the non zero structure of $A$, but random vectors cannot 
  exploit this induced structure of $A^{-1}$.
Probing is a technique that, given a sparsity pattern of $A$, 
  discovers elements of $A$ through matrix-vector multiplications 
  with specially designed vectors.
In the case of $A^{-1}$, the pattern is obtained by distance-$k$ 
  coloring of the graph of $A$.
For sufficiently large $k$, the method produces accurate trace estimates
  but the cost of producing the colorings becomes prohibitively expensive.
More importantly, it is difficult to search for an optimal $k$ value, 
  since none of the work for prior choices of $k$ can be reused.

\cred
First, we introduce the idea of hierarchical probing that produces 
  distance-$2^i$ colorings for a sequence of distances
  $2^0, 2^1, \ldots , 2^m$, up to the diameter of the graph.
To achieve this, we do not color the entire graph, but at each level, $i$, 
  we compute the distance-1 coloring independently for each of the 
  node-groups associated with a color of the distance-$(2^{i-1})$ coloring.
Second, based on this idea, we develop an algorithm for uniform, toroidal 
  lattices that simply applies bit-arithmetic on local coordinates 
  to produce the hierarchical permutation.
Third, we provide an algorithm for choosing an appropriate sequence 
  of Hadamard and Fourier vectors, so that earlier vectors in the sequence 
  correspond to hierarchical probing vectors of smaller distances.
This allows us to increase the number of probing vectors until 
  the required accuracy is achieved.

Several experiments show that when a decay structure exists in the matrix,
  our algorithm finds it and approximates the trace incrementally, starting 
  with the most important contributions.
We have observed up to an order of magnitude speedup over the standard 
  Monte Carlo. 

\end{abstract}
\pagestyle{myheadings}
\thispagestyle{plain}
\markboth{Stathopoulos, Laeuchli, Orginos}{Hierarchical probing for trace estimation}
\section{Introduction}

A computationally challenging task in numerical linear algebra is the 
  estimation of the trace or the determinant of functions of matrices. 
The source of computational difficulty is twofold: first, the function of 
  the matrix or its action on a vector must be computed, and following this, 
  its trace must be estimated. 
We focus on the trace of the inverse of the matrix, $\TrAi$, but our 
  approach can be adapted to the determinant, as $\mathbf{Tr}({\log A})$, 
  or other functions.

These problems are common in many statistical applications 
  \cite{Hutchinson_90}, in data mining \cite{Bekas_diagonal},
  in uncertainty quantification \cite{Bekas_Uncertain},
  in optimal code design \cite{Optimal_design},
  as well as in quantum physics applications such as quantum Monte Carlo
  \cite{Eric_ratioDets}.
Our motivation for this work comes from some particularly computationally 
  intensive problems in lattice quantum chromodynamics (LQCD).
The goal of LQCD is to calculate the properties, structure, and interactions  
  of hadrons, the basic constituents of matter \cite{Gupta:1998,Rothe:2005}.
These calculations, together with experiments from particle accelerators,
  enable physicists to form a comprehensive picture of the subatomic world.
Computation of observables in LQCD entails averaging of correlation 
  functions over an ensemble of gauge fields, which in turn requires
  certain sections of the inverse of a matrix, the trace of the inverse, 
  or sometimes the ratio of determinants of two such matrices. 
The matrix discretizes a differential operator containing
  first and second order derivatives on a four dimensional 
  (space-time), uniform lattice with periodic conditions, 
  where each point (site) has an internal dimension of 12 related to the
  color and spin degrees of freedom.
Clearly, this is a very large and sparse matrix.

For problems of small size, $\TrAi$ can be computed efficiently 
  through variations of dense or sparse LU decomposition methods, 
  \colorred{
  \cite{Amestory_inverseEntries,Duff:1989}.
}
For larger problems stemming from some 2-D and 3-D discretizations, 
  algorithms have been developed using the banded structure, or based
  on preconditioning-type algorithms such as domain decomposition
\colorred{
  \cite{Lin-Car_traceInv, Lin-Chao_traceInv, Tang_Saad_traceInv}.
}
Such algorithms, however, focus on the accurate computation of the trace,
  and have complexities ${\cal O}(N^{3/2})$ or ${\cal O}(N^2)$ for 2-D or
  3-D problems respectively, where $N$ is the matrix size. 
\colorred{
For our 4-D LQCD problem, time and memory complexities are even higher,
  making such methods impractical for large $N$.
}
Moreover, this expense is not warranted in many problems, including LQCD,
  where only a low accuracy estimation of the trace is required.
Preconditioner-based approaches have also been proposed for related problems 
\cite{Ipsen_Lee_TR_determinants,Ipsen_Lee_zoneDeterminant,Reusken:2002:ApproxDeterm,Tang_Saad_traceInv} but their computational costs tend to be similar 
  to the above approaches.
\colorred{
On the other hand, the matrix trace, which can be viewed as an expectation 
  value (sum of all eigenvalues), is not approximated well by a projection 
  onto a Krylov space.
The Lanczos method, even in exact arithmetic, would require $N$ 
  iterations to calculate $\TrAi$. 
}

In the realm of such problem sizes, stochastic methods have become the 
  standard approach.
Given $f$ some function of $A$ and $s$ random vectors $z_k$ with 
  independent, identically distributed (i.i.d.) entries, 
  $\frac{1}{s} \sum_{k=1}^s z_k^T f(A)z_k$
 is an unbiased estimator of $\Tr(f(A))$ \cite{Hutchinson_90}.
Thus, we can set up a Monte Carlo process (MC) and at each step 
  approximate the Gaussian quadrature $z^T f(A)z$ using the Lanczos method
  \cite{Bai_Fahey_Golub,HongbinGuo_Renaut_Bilinear,Strakos_bilinear}.
MC converges slowly, as $O(1/\sqrt{s})$, but it is unbiased
  which is important in LQCD where the traces of a large sequence of
  matrices must be averaged.

The convergence rate of the MC method depends on the variance of the 
  estimator.
Choosing the vector entries as \colorred{$\mathbb{Z}_2$} noise 
  (i.e., $z(i) = \pm 1$ with probability $0.5$) is known to minimize variance
  over real random vectors \cite{Bernardson:1993yg,Hutchinson_90}.
Various choices of $\mathbb{Z}_2$ vectors have been proposed 
  \cite{Bekas_diagonal, Iitaka_trace, traceHongKong}.
In \cite{Toledo_trace}, the authors provide a thorough analysis of both 
  the variance and the number of steps required to achieve an 
  $\epsilon$-approximation of the trace for various choices of vectors.
They show that differences exist in the theoretical bounds governing 
  convergence, but their experiments do not reveal these differences in 
  practice.
As we can rarely afford to run MC to the asymptotic regime, 
  practical improvements can only come from deterministic variance 
  reduction mechanisms.  
Over the last decade there have been some efforts to equip MC with 
  specially selected vectors based on the structure of the matrix 
\cite{ Toledo_trace, Bekas_diagonal, HongbinGuo_traceFuncs, Tang_Saad_probing,
traceHongKong}.
In this research, we address the problems that prevented two of these 
  otherwise very promising methods from capturing the structure 
  of more general matrices.

The first method borrows ideas from coding theory and selects deterministic 
  vectors for the MC as columns of a Hadamard matrix \cite{Bekas_diagonal}.
These vectors are orthogonal and, although they produce the exact answer
  in $N$ steps, their benefit stems from systematically capturing
  certain diagonals of the matrix. 
For example, if we use the first $2^m$ Hadamard vectors,
  the error in the trace approximation comes only from non-zero elements on the 
  ($k2^m$)th matrix diagonal, $k=1,\ldots, N/2^m$. 
Thus, the MC iteration continues annihilating more diagonals with more 
  Hadamard vectors, until it achieves the required accuracy.
However, in most practical problems the matrix bandwidth is too large, 
  the non-zero diagonals do not fall on the required positions,
  or the matrix is not even sparse (which is typically the case for $A^{-1}$).

The second method is based on probing \cite{Tang_Saad_probing}, which 
  selects vectors that annihilate the error contribution from the
  heaviest elements of $A^{-1}$.
For a large class of sparse matrices, elements of $A^{-1}$ decay 
  exponentially away from the non-zero elements of $A$.
\colorred{
In other words, the magnitude of the $A^{-1}_{i,j}$ element relates to 
 the distance of the minimum path between nodes $i$ and $j$ in the graph of $A$.
}
Assume that the graph of $A$ has a distance-$k$ coloring 
  (or distance-1 coloring of the graph of $A^k$) with $m$ colors.
Then, if we define the vectors $z_j, j=1,\ldots, m$, 
  with $z_j(i) = 1$ if color($i$)=$j$, and $z_j(i) = 0$ otherwise, 
  we obtain $\Tr(A) = \sum_{j=1}^m z_j^TAz_j$.  
For $\TrAi$ the equation is not exact, \colorred{ 
  but it annihilates errors from
  all elements of $A^{-1}$ that correspond to paths between vertices 
  that are distance-$m$ neighbors in $A$.
}
The probing technique has been used for decades in the context
  of approximating the Jacobian matrix \cite{CoMo83a,Pothen_Coloring} 
  or other matrices \cite{Seifert_probing}.
Its use for approximating the diagonal of $A^{-1}$ in 
  \cite{Tang_Saad_probing} (see also \cite{Bekas_diagonal}) 
  is promising as it selects the important areas of $A^{-1}$ rather
  than the predetermined structure dictated by Hadamard vectors.
\colorred{
However, the accuracy of the trace estimate obtained through 
  a specific distance-$k$ probing can only be improved by applying 
  Monte Carlo, using random vectors that follow the structure 
  of each probing vector.
To take advantage of a higher distance probing, all
  previous work has to be discarded, and the method rerun for a larger $k$.
We discuss this in Sections \ref{subsection:probing} and \ref{section:HP}.
}

In this paper we introduce hierarchical probing which avoids the problems 
  of the previous two methods; namely it annihilates error stemming from 
  the heaviest parts of $A^{-1}$, and it does so in an incremental way 
  until the required accuracy is met. 
\colorred{
To achieve this, we relax the requirement of distance-$k$ coloring of 
  the entire graph.
The idea is to obtain recursively a (suboptimal) distance-$2^{i+1}$ coloring 
  by independently computing distance-1 colorings of the subgraphs 
  corresponding to each color from the distance-$2^i$ coloring.
The recursion stops when all the color-subgraphs are dense, i.e., 
  we have covered all distances up to the diameter of the graph.
We call this method, ``hierarchical coloring''.
The number of colors produced may differ between subgroups for general 
  matrices.
For lattices, however, each subgroup has the same number of colors,
  which enables an elegant, hierarchical basis for probing.
If we consider an appropriate ordering of the Hadamard and/or Fourier vectors
  and permute their rows based on the hierarchical coloring,
  the first $m$ such vectors constitute a basis for the 
  corresponding $m$ probing vectors.
We call this method, ``hierarchical probing''.
It can be implemented using only bit arithmetic, independently on 
  each lattice site.
}
We also address the issue of statistical bias by viewing hierarchical probing 
  as a method to create a hierarchical basis starting from any vector, 
  including random.

Section 2 describes the general idea of hierarchical coloring and probing.
In Section 3, we consider the case of uniform grids and tori 
  with sizes that are power of two, and develop a modified hierarchical 
  probing that uses only local coordinate information and bit operations 
  and produces the hierarchical probing vectors in parallel and 
  highly efficiently.
In Section \ref{sec:nopower2}, we extend the hierarchical coloring 
  to lattices with size that includes non-power of two factors.
In Section 4, we provide several experiments for typical lattices and 
  problems from LQCD that show that MC with hierarchical probing 
  has much smaller variance than random vectors and performs 
  equally well or better than the expensive, large distance probing method.

\subsection{Preliminaries}
We use vector subscripts to denote the order of a sequence of vectors,
  and parentheses to denote the index of the entries of a vector.
We use MATLAB notation to refer to row or column numbers and ranges.
The matrix $A$, of size $N\times N$, is assumed to have a
  symmetric structure (undirected graph).

\subsubsection{Lattice QCD problems}
Lattice Quantum Chromo-Dynamics  (LQCD) is a formulation of  
  Quantum Chromo-Dynamics (QCD) that allows for numerical calculations of 
  properties of strongly interacting matter 
  (Hadron Physics)~\cite{Wilson:1974sk}.  
These calculations are performed through Monte Carlo
  computations of the discretized theory on a finite 4 dimensional Euclidean lattice.  Physical results are obtained after extrapolation of the lattice spacing to zero. Hence calculations on multiple lattice sizes are required for taking the continuum and infinite volume limits. 
  In this formulation, a large sparse matrix $D$ called the Dirac
  matrix plays a central role. This matrix depends explicitly on the gauge fields $U$. The physical observables in a LQCD calculation
  are computed as averages over the ensemble of gauge field configurations.
    In various stages of the computation one needs, among other things,  to estimate the determinant as well as the trace of the inverse of this matrix. 
  The dimensionality of the matrix is $3\times 4 \times L_s^3 \times L_t$, where $L_s$ and $L_t$ are the dimensions of the spatial and temporal directions
  of the space-time lattice,   3 is the dimension of an internal space named ``color'', and 4 is the dimension of the space associated with the spin and particle/antiparticle degrees of freedom.
  Typical lattice sizes in todays calculations have $L_s=32$ and $L_t=64$ and the largest calculations performed on leadership class machines
  at DOE or NSF supercomputing centers have $L_s=64$ and $L_t=128$. As computational resources become available and precision requirements grow, lattices sizes will become even bigger.

\subsubsection{The Monte Carlo method for $\TrAi$}
Hutchinson introduced the standard MC method for estimating the trace of a 
 matrix and proved the following \cite{Hutchinson_90}.
\begin{lemma}
\label{lemma:MCvar}
Let $A$ be a matrix of size $N\times N$ and denote by $\tilde A = A-diag(A)$.
Let $z$ be a $\mathbb{Z}_2$ random vector (i.e., whose entries are i.i.d
Rademacher random variables $Pr(z(i) = \pm 1) = 1/2$).
Then, $z^TAz$ is an unbiased estimator of $\Tr(A)$, i.e.,
$$E(z^TAz) = \Tr(A),$$
and
$$ var(z^TAz) = \|\tilde A\|_F^2 = 
		2\left(\|A\|_F^2 - \sum_{i=1}^N A(i,i)^2\right).$$
\end{lemma}
The MC method converges with rate $\sqrt{var(z^TAz)/s}$, where
  $s$ is the sample size of the estimator (number of random vectors). 
Thus, the MC converges in one step for diagonal matrices, 
  and very fast for strongly diagonal dominant matrices. 
More relevant to our $\TrAi$ problem is that large off-diagonal elements 
  of $A^{-1}$ contribute more to the variance $\| \widetilde{A^{-1}}\|_F^2$ 
  and thus to slower convergence.

Computationally, the Gaussian quadrature $z_kA^{-1}z_k$ can be computed 
  using the Lanczos method 
  \cite{Bai_Fahey_Golub,Golub_meurant_moments,Golub_Strakos_moments}.
This method produces also upper and lower bounds on the quadrature,
  which are useful for terminating the process.
A simpler alternative is to solve the linear system $A^{-1}z_k$.
Although this is not recommended for non-Hermitian systems because
  of worse floating point behavior \cite{Strakos_bilinear}, 
  for Hermitian systems it can be as effective if we stop the system
  earlier. 
Specifically, the quadrature error in Lanczos converges as the square 
  of the system residual norm \cite{Golub_meurant_moments},
  and therefore we only need to let the residual converge to the 
  square root of the required tolerance.
A potential advantage of solving $A^{-1}z_k$ is that the result can be 
  reused when computing multiple correlation functions involving 
  bilinear forms $y^T A^{-1}z_k$ (e.g., in LQCD).

\subsubsection{Probing}
\label{subsection:probing}
Probing has been used extensively for the estimation of 
  sparse Jacobians \cite{CoMo83a,Pothen_Coloring}, for preconditioning
  \cite{Seifert_probing}, and in Density Functional Theory 
  for approximating the diagonal of a dense projector whose elements
  decay away from the main diagonal \cite{Bekas_diagonal, Tang_Saad_probing}.
The idea is to expose the structure and recover the non-zero entries of 
  a matrix by multiplying it with a small, specially chosen set of vectors.
For example, we can recover the elements of a diagonal matrix through
  a matrix-vector multiplication with the vector of $N$ 1's, 
  $\onesN = [1, \ldots, 1]^T$.
Similarly, a banded matrix of bandwidth $b$ can be found by 
  matrix-vector multiplications with vectors $z_k, k=1,\ldots ,b$, where 
$$z_k(i) = \left\{
\begin{array}{ll}
1, & \mbox{for\ } i=k:b:N \\
0, & \mbox{otherwise}
\end{array}\right. .$$

To find the trace (or more generally the main diagonal) of a matrix, 
  the methods are based on the following proposition \cite{Bekas_diagonal}.
\begin{proposition}
\label{prop:structuralOrtho}
Let $Z\in \Re^{N\times s}$ be the matrix of the $s$ vectors used in the 
  MC trace estimator.
If the $i$-th row of $Z$ is orthogonal to all those rows $j$ of $Z$
  for which $A(i,j) \neq 0$, then the trace estimator yields the
  exact $\Tr(A)$.
\end{proposition}

In the above example of a banded matrix, we choose the vectors $z_k$ 
  such that their rows only overlap for structurally orthogonal rows of $A$ 
  (i.e., for rows farther than $b$ apart). 
Thus the proposition applies and the trace computed with these $z_k$ is exact.

If $A$ is not banded but its sparsity pattern is known, graph coloring 
  can be used to identify the structurally orthogonal segments of rows, 
  and derive the appropriate probing vectors \cite{Tang_Saad_probing}.
Assume the graph of $A$ is colorable with $m$ colors,
  each color having $n(k)$ number of vertices, $k=1,\ldots, m$.
The coloring is better visualized if we let $q$ be the permutation vector 
  that orders first vertices of color 1, then vertices of color 2, and so on.
Then $A(q,q)$ has $m$ blocks along the diagonal, the $k$-th block is of 
  dimension $n(k)$, and each block is a diagonal matrix.
Figure~\ref{fig:ColorPerm} shows an example of the sparsity structure of
  a permuted 4-colorable matrix.
Computationally, permuting $A$ is not needed.
If we define the vectors:
\begin{equation}
\label{eq:colorProbingVecs}
z_k(i) = \left\{ \begin{array}{ll}
		1 & \mbox{if } color(i)=k\\
		0 & \mbox{otherwise}\\
		\end{array}\right., k=1,\ldots, m
\end{equation}
  we see that Proposition~\ref{prop:structuralOrtho} applies, and therefore
  $\Tr(A) = \sum_{k=1}^m z_k^TAz_k$. 

\begin{figure}[ht]
\centering
\includegraphics[width=0.5\textwidth]{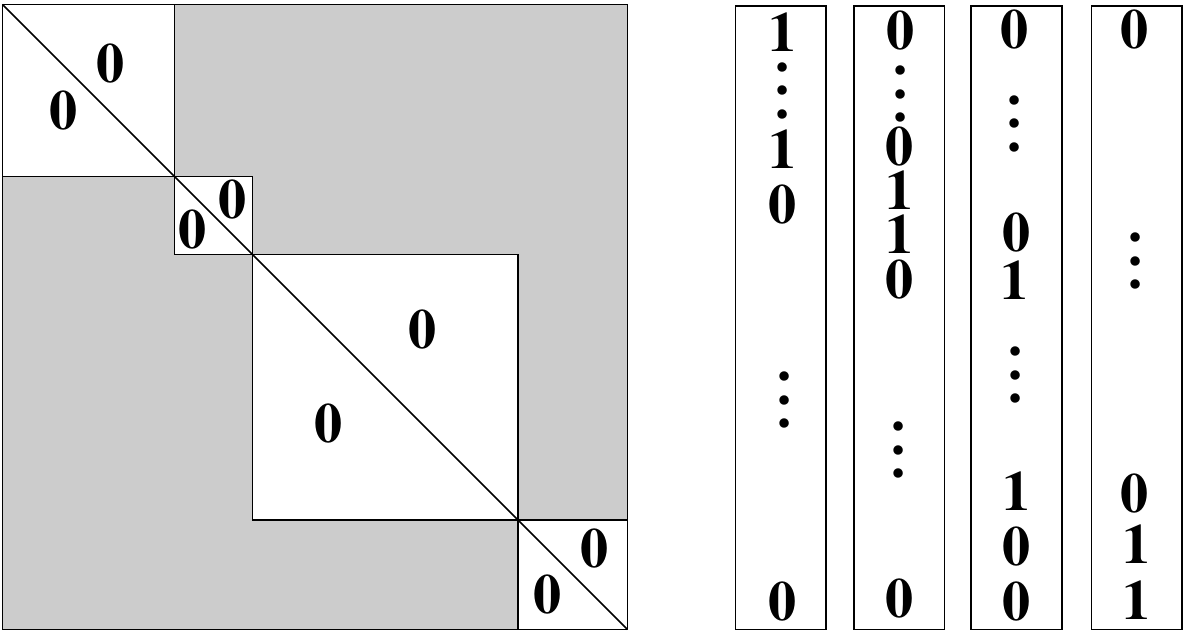}
\caption{Visualizing a 4-colorable matrix permuted such that 
	all rows corresponding to color 1 appear first, for color 2 appear 
  	second, and so on. Each diagonal block is a diagonal matrix. The four 
  	probing vectors with 1s in the corresponding blocks are shown 
	on the right.}
\label{fig:ColorPerm}
\end{figure}

When the matrix is dense and all its elements are of similar magnitude,
  there is no structure to be exploited by probing. 
The inverse of a sparse matrix is typically dense, but, 
  for many applications, its elements decay on locations that are 
  farther from the locations of the non-zero elements of $A$.
Such small elements of $A^{-1}$ can be dropped, and the remaining $A^{-1}$ 
  is sparse and thus colorable.
\colorred{
Diagonal dominance of the matrix is a sufficient (but not necessary)
  condition for the decay to occur \cite{CoMo83a,Tang_Saad_probing}.
}
This property is exploited by approximate inverse preconditioners
  and can be explained from various points of view, including
  Green's function for differential operators, the power series 
  expansion of $A^{-1}$, or a purely graph theoretical view
  \cite{ Benzi_Boito_Razouk_decay, Benzi_Golub_decay, EC_YS_97, THuckle_99}.
In the context of probing, we drop elements $A^{-1}(i,j)$ whose
  vertices $i$ and $j$ are farther than $k$ links apart in the graph of $A$.
Because this graph corresponds to the matrix $A^k$, 
  our required distance-$k$ coloring is simply the distance-1 coloring 
  of the matrix $A^k$ \cite{Pothen_Coloring,Tang_Saad_probing}.
Computing $A^k$ for large $k$, however, is time and/or memory intensive.

The effectiveness of probing depends on the decay properties of the 
  elements of $A^{-1}$, and the choice of $k$ in the distance-$k$ coloring. 
The problem is that $k$ depends both on the structure and the numerical
  properties of the matrix. 
If elements of $A^{-1}$ exhibit slow decay, choosing $k$ too small does not 
  produce sufficiently accurate estimates 
  because large elements of $A^{-1}$ (linking vertices that are farther 
  than $k$ apart) contribute to the variance in Lemma \ref{lemma:MCvar}.
Choosing $k$ too large increases the number of quadratures unnecessarily, 
  and more importantly, makes the coloring of $A^k$ prohibitive.
This problem has also been identified in \cite{Tang_Saad_probing} but 
  no solution proposed.

\colorred{
A conservative approach is to use probing for a small distance
  (typically 1 or 2) to remove the variance associated only with the
  largest, off-diagonal parts of the matrix.
Then, for each of the resulting $m$ probing vectors, we generate $s$ 
  random vectors that follow the non-zero structure of the corresponding
  probing vector, and perform $s$ MC steps (requiring $m s$ quadratures).
}
In LQCD this method is called dilution. 
In its most common form it performs a red-black ordering on the uniform
  lattice and uses the MC estimator to compute two partial traces:  
  one restricted on the red sites, the other on the black sites of the lattice 
  \cite{Bali:2009hu,Foley:2005ac,Morningstar_Peardon_etal_2011}.
Therefore, all variance caused by the direct red-black connections of 
  $A^{-1}$ is removed. 
The improvement is modest, however, so additional ``dilution''
  is required based on spin-color coordinates
	\cite{Babich:2011,Morningstar_Peardon_etal_2011}.

\subsubsection{Hadamard vectors}
An $N\times N$ matrix $H$ is a Hadamard matrix of order $N$ if it 
  has entries $H(i,j) = \pm 1$ and $HH^T = N I$, where $I$ is the identity
  matrix of order $N$ \cite{Bekas_diagonal,Hadamard_matrices}.
$N$ must be 1, 2, or a multiple of 4.
We restrict our attention to Hadamard matrices whose order is a power of 2,
  and can be recursively obtained as: 
$$H_2 = \left[ \begin{array}{cc}
	1 & 1\\
	1 & -1\\
	\end{array} \right],
\ \  H_{2n} = \left[ \begin{array}{cc}
				H_n & H_n\\
				H_n &-H_n\\
			    \end{array} \right] = H_2 \otimes H_{n}.$$ 
For powers of two, $H_n$ is also symmetric, and its elements can be obtained 
  directly as 
\begin{equation}
\label{eq:HadamardElements}
H_n(i,j) = (-1)^{\sum_{k=1}^{\log N} i_kj_k}, 
\end{equation}
  where $(i_{\log N},\ldots,i_1)_2$
  and $(j_{\log N},\ldots,j_1)_2$ are the binary representations of 
  of $i-1$ and $j-1$ respectively.
We also use the following notation to denote Hadamard columns (vectors):
  $h_j = H_n(:,j+1), j=0,\ldots, n-1$.
Hadamard matrices are often called the integer version of the discrete 
  Fourier matrices,
\begin{equation}
\label{eq:FourierElements}
F_n(j,k) = e^{2\pi (j-1) (k-1) \sqrt{-1}/n}.
\end{equation}
For $n=2$, $H_2 = F_2$, but for $n>2$, $F_n$ are complex.
These matrices have been studied extensively in coding theory where
  the problem is to design a code (a set of $s<N$ vectors $Z$) for which
  $ZZ^T$ is as close to identity as possible \cite{Bekas_diagonal}.
$H_n$ and $F_n$ vectors satisfy the well known Welch bounds but the
  $H_n$ do not achieve equality \cite{Welch_bounds}.
Moreover, $F_n$ are not restricted in powers of two.
Still, Hadamard matrices involve only real arithmetic, 
  which is important for efficiency and interoperability with real codes, 
  and it is easy to identify the non-zero pattern they generate.
Later, we will view the Hadamard matrix as a methodical way
  to build an orthogonal basis of a power of two size space.

Consider the first $2^k$ columns of a Hadamard matrix $Z=H(:,1:2^k)$.
The non-zero pattern of the matrix $ZZ^T$ consists of the $i2^k$ upper 
  and lower diagonals, $i=0,1,\ldots$ \cite{Bekas_diagonal}. 
Because $\Tr(Z^TA^{-1}Z) = \Tr(A^{-1}ZZ^T)$ and 
  because of Lemma~\ref{lemma:MCvar} and Proposition \ref{prop:structuralOrtho},
  the error in the MC estimation of the trace is induced only by 
  the off-diagonal elements of $A^{-1}$ that appear on the same locations
  as the non-zero diagonals of $ZZ^T$.
If the matrix is banded or its diagonals do not coincide with the ones 
  of $ZZ^T$, the trace estimation is exact.
When the off-diagonal elements of $A^{-1}$ decay exponentially away 
  from the main diagonal, increasing 
  the number of Hadamard vectors achieves a consistent (if not monotonic)
  reduction of the error.
We note that this special structure of $ZZ^T$ is achieved only when the
  number of vectors, $s$, is a power of two.
For $2^k < s < 2^{k+1}$, the structure of $ZZ^T$ is dense in general, 
  but the weight of $ZZ^T$ elements is largest on the main diagonal 
  (equal to $s$) and decreases between diagonals $i2^k$ and $(i+1)2^k$.
Thus, estimation accuracy improves with $s$, even for dense matrices.
However, to annihilate a certain sparsity structure of a matrix, 
  the estimates at only $s=2^k$ should be considered.
Similar properties apply for the $F_n$ matrices.

\section{Hierarchical probing}
\label{section:HP}
We seek to construct special vectors for the MC estimator 
  that perform at least as well as $\mathbb{Z}_2$ noise vectors, 
  but can also exploit the structure of the matrix, 
  when such structure exists.
Although Hadamard vectors seem natural for banded matrices, 
  they do not take into account deviations from the expected
  structure.
For example, the first two Hadamard vectors compute the exact 
  trace of a tridiagonal matrix. 
For the matrix that corresponds to a 2-D uniform lattice of size 
  $2^n\times 2^n$ with periodic boundary conditions and lexicographic
  ordering, producing the exact trace requires the first $s=2^{n+1}$ 
  Hadamard vectors.
However, if we consider the red-black ordering of the same matrix, 
  only two Hadamard vectors, the first $h_0$ and the middle $h_{2^{n-1}}$,
  are sufficient. 
This is shown in Figure \ref{fig:red-black-hada}.

\begin{figure}[t]
\centering
\begin{minipage}{0.2in}
\tiny
\begin{verbatim}
1     1
1    -1
1     1
1    -1
1     1
1    -1
1     1
1    -1
1     1
1    -1
1     1
1    -1
1     1
1    -1
1     1
1    -1
\end{verbatim}
\end{minipage}
\hspace{0.1in}
\begin{minipage}{3in}
\includegraphics[width=\textwidth,angle=0]{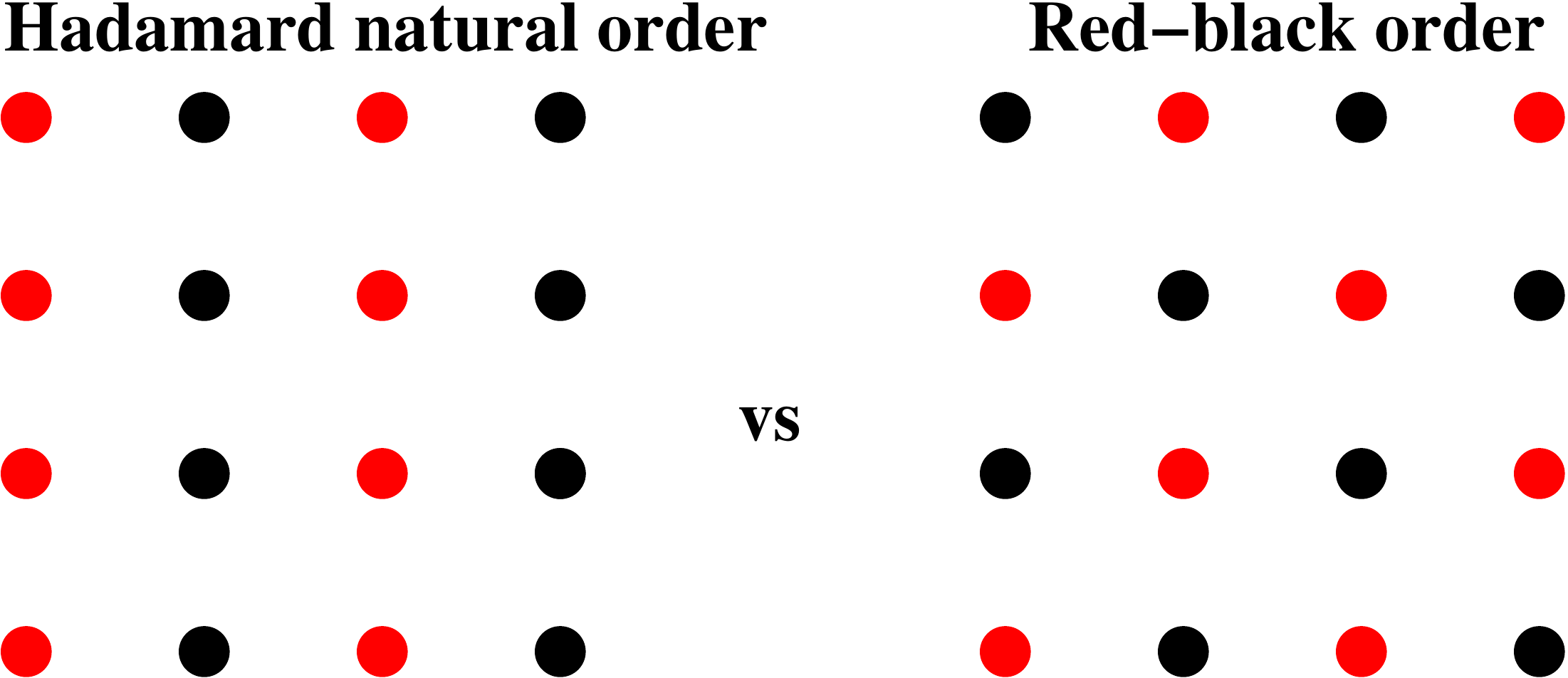}
\end{minipage}\hspace{0.1in}
\begin{minipage}{0.2in}
\tiny
\begin{verbatim}
1     1
1     1
1     1
1     1
1     1
1     1
1     1
1     1
1    -1
1    -1
1    -1
1    -1
1    -1
1    -1
1    -1
1    -1
\end{verbatim} 
\end{minipage}
\caption{Left: the first two natural order Hadamard vectors select some 
  neighboring vertices in the lexicographic ordering of a 2-D uniform lattice
  and thus cannot cancel their variance.
  Right: if the grid is permuted with the red nodes first, the first and 
  the middle Hadamard vectors completely cancel variance from nearest
  neighbors and correspond to the distance-1 probing vectors.}
\label{fig:red-black-hada}
\end{figure}

The previous example shows that although Hadamard vectors are a useful 
  tool, probing is the method that discovers matrix structure.
Therefore, we turn to the problem of how to perform probing efficiently 
  on $A^k$ and for large $k$.
Ideally, a method should start with a small $k$ and increase it until 
  it achieves sufficient accuracy.
However, the colorings, and therefore the probing vectors, for two 
  different $k$'s are not related in general.
Thus, in addition to the expense of the new coloring, all the quadratures 
  performed previously have to be discarded and new ones performed.
Our hierarchical probing provides an elegant solution to this problem.

First let us persuade the reader that work from a previous distance-$k$
  probing cannot be reused in general.
Assume the distance-1 coloring of a matrix of size 6 produced three colors:
  color 1 has rows 1 and 2,
  color 2 has rows 3 and 4,
  color 3 has rows 5 and 6.
Next we perform a distance-2 coloring of $A$, and assume there are four colors:
  color 1 has row 1, 
  color 2 has rows 2 and 3, 
  color 3 has rows 4 and 5, 
  color 4 has row 6.
As in Figure \ref{fig:ColorPerm}, the distance-1 and distance-2 
   probing vectors, $Z^{(1)}$ and $Z^{(2)}$ respectively, are the following:
$$
Z^{(1)} = \left[
  \begin{array}{ccc}
   1 & 0 & 0 \\
   1 & 0 & 0 \\
   0 & 1 & 0 \\
   0 & 1 & 0 \\
   0 & 0 & 1 \\
   0 & 0 & 1 \\
\end{array} \right],\
Z^{(2)} = \left[
  \begin{array}{cccc}
   1 & 0 & 0 & 0\\
   0 & 1 & 0 & 0\\
   0 & 1 & 0 & 0\\
   0 & 0 & 1 & 0\\
   0 & 0 & 1 & 0\\
   0 & 0 & 0 & 1\\
\end{array} \right].
$$
Unfortunately, the three computed quadratures $Z^{(1)T}A^{-1}Z^{(1)}$
  (or solutions to $A^{-1}Z^{(1)}$) cannot be used to avoid 
  recomputation of the four quadratures $Z^{(2)T}A^{-1}Z^{(2)}$.

Consider now a matrix of size 8 with two colors in its distance-1 coloring.
Assume that its distance-2 coloring produces four colors, and that all 
  rows with the same color belong also to the same color group for 
  distance-1.
Then the subspace of the corresponding probing vectors is spanned
  by certain Hadamard vectors:
\begin{equation*}
\begin{array}{ll}
Z^{(1)} = 		&			Z^{(2)} = \\
\left[ \begin{array}{cc}
   1 & 0 \\
   1 & 0 \\
   1 & 0 \\
   1 & 0 \\
   0 & 1 \\
   0 & 1 \\
   0 & 1 \\
   0 & 1 \\
\end{array} \right] 
\in span(
\left[ \begin{array}{rr}
1  &  1\\
1  &  1\\
1  &  1\\
1  &  1\\
1  & -1\\
1  & -1\\
1  & -1\\
1  & -1\\
\end{array} \right] 
),\
& 
\left[ \begin{array}{cccc}
   1 & 0 & 0 & 0\\
   1 & 0 & 0 & 0\\
   0 & 1 & 0 & 0\\
   0 & 1 & 0 & 0\\
   0 & 0 & 1 & 0\\
   0 & 0 & 1 & 0\\
   0 & 0 & 0 & 1\\
   0 & 0 & 0 & 1\\
\end{array} \right] 
\in span(
\left[ \begin{array}{rrrr}
 1  &  1  &  1  &  1 \\
 1  &  1  &  1  &  1 \\
 1  &  1  & -1  & -1 \\
 1  &  1  & -1  & -1 \\
 1  & -1  &  1  & -1 \\
 1  & -1  &  1  & -1 \\
 1  & -1  & -1  &  1 \\
 1  & -1  & -1  &  1 \\
\end{array} \right] 
).
\end{array}
\end{equation*}
The four Hadamard vectors are $h_0, h_4, h_2, h_6$.
More interesting than the equality of the spans is that the two bases
  are an orthogonal transformation of each other.
Specifically, $Z^{(1)} = 1/2 [h_0, h_4] H_2$
   and $Z^{(2)} = 1/2 [h_0, h_2, h_4, h_6] H_4$.
Because the trace is invariant under orthogonal transformations, 
  we can use the Hadamard vectors instead (as we implicitly did 
  in Figure \ref{fig:red-black-hada} for the lattice).
Clearly, for this case, the quadratures of the first two vectors
  can be reused so that the distance-2 probing will need computations for 
  only two additional vectors.
 
A key difference between the two examples is the nesting of colors
  between successive colorings.
In general, such nesting cannot be expected and thus an incremental 
  approach to probing will necessarily discard prior work.
A second difference is that all color groups are split into 
  the same number of colors in the successive coloring.
\colorred{
This property holds for lattices.
}
Our hierarchical probing method performs two tasks. 
First, it enforces a nested, all-distance coloring by making certain 
  compromises.
Second, it finds a set of Hadamard vectors that spans exactly 
  the probing basis for the above nested coloring 
  (i.e., the $Z^{(1)}, Z^{(2)},\ldots$ of the previous subsection).
\colorred{
When the matrix size is not a power of two, 
  Fourier vectors can be used instead.
}

\subsection{The hierarchical coloring algorithm}
\label{sec:hierPerm}
The goal is to provide a permutation of all matrix rows that 
  corresponds to a nested coloring for each probing distance.
Algorithm \ref{HPtask1} is based on the following ideas:
\begin{itemize}
\item We perform probing only for distance-$2^k$ colorings, but for 
  	all $k=0,1,\ldots$ until the number of colors is $N$ 
	(i.e., a dense graph).

\item As it is infeasible to color $A^{2^k}$ for all possible $k$,
  we perform hierarchical coloring independently on each submatrix 
  that corresponds to the row indices of a each color $i$. 
  The recursion continues until every element of the matrix is colored uniquely.

\item \colorred{
  We use any efficient greedy coloring algorithm (see \cite{Saad_linear}),
  and permute nodes of the same color together using Colors2Perm().
}
\end{itemize}

\begin{algorithm}
\caption{ \label{HPtask1} Hierarchical coloring}
\begin{tabbing}
xxx\=xxx\=ixx\=ixx\=xxxxxxxxxxx\=\kill
\% \> {\bf Input:} \\
\% \> Matrix block: Ablock. The first time this is the original matrix.\\
\% \> start, end: The global range of row indices of the current block. \\
\% \> n: The recursive level.\\
\% \> {\bf Input/Output:}\\
\% \> perm: Global permutation array reflecting the hierarchical coloring.\\
\% \>\>\> Update perm(start:end).\\
\% \> colorOffsets: Global array, initially empty. 
	Records (start, end, n) for all blocks\\
1.\> {\bf function} {\sc HierarchicalColoring}(Ablock, start, end, n) \\
2. \>\>{\bf global} colorOffsets 
		= [ colorOffsets ; (start, start+size(Ablock,1)-1, n) ]\\
3. \> \> [numColors newColors] = GreedyColor(Ablock)\\
4. \> \> {\bf if} (numColors=end-start+1) OR (numColors=1) {\bf then}  \\
5. \> \>\>  {\bf return} perm	\>\>\% dense or diagonal matrix. Do not reorder global rows \\
6. \> \> [perm(start : end) numRowsPerColor] = Colors2Perm(newColors)\\
7. \> \> colorStarts = start\\
8. \> \> {\bf for} (i = 1 : numColors)\\
9. \> \>\> colorEnds = colorStarts+numRowsPerColor(i)-1\\
10. \> \>\> rowIndices = perm(colorStarts : colorEnds)\\
11. \> \>\> AdiagBlock = Ablock(rowIndices, : )*Ablock( : , rowIndices)\\
12. \> \>\> HierarchicalColoring(AdiagBlock,colorStarts,colorEnds,n+1)
\end{tabbing}
\end{algorithm}
\colorred{
We initialize perm to the identity permutation, colorOffsets to empty, and call the algorithm at level 1 with 
  the original matrix $A$: {\sc HierarchicalColoring}($A,1,N,1$).
}
In line 2, we record (start,end) indices and recursion level for this block.
At the end, these offsets provide the number of nodes per color at every level.
In line 3, the local matrix is colored. 
Lines 4-5 are the base case of the recursion. 
We stop when the given subgraph is fully connected (could be just one node),
  or it is fully disconnected and is assigned one color.
In either case, the existing permutation is a valid ordering.
In line 6, the algorithm Colors2Perm() creates the local coloring 
  permutation array from the given local colors.
Note that this only reorders elements appearing between $start$ and $end$
  in the permutation array. 
Moreover, Colors2Perm() returns an array with the number of rows assigned
  to each color.
In line 8, we visit each color of this local matrix, and in line 11, 
  we compute $A^2$ {\em but only for the row indices of this color}.
\colorred{
This is the compromise step, where each subgraph is colored independently
  in the following levels.
}
Line 12 calls the algorithm recursively to produce a hierarchical coloring 
  only for the rows of the current color $i$.

\colorred{
Hierarchical coloring produces more colors at distance $2^k$ than 
  classic coloring of the graph of $A^{2^k}$
  (although the classic method would be computationally prohibitive 
  even for small $k$).
If the task were to approximate the trace of the matrix $A^{2^k}$, 
  the extra colors would be redundant and the additional probing 
  vectors would represent unnecessary computational work.
However, we approximate the trace of $A^{-1}$, which is dense.
Thus, the larger number of hierarchical probing vectors at distance $2^k$
  will also approximate some elements that represent
  node distances larger than $2^k$.
}

Line 11 performs a sparse matrix-matrix multiplication and could 
  result in a much denser matrix, AdiagBlock.
However, depending on the number of colors at the level, 
  AdiagBlock is of much smaller dimension than Ablock, 
  and only one such block for color $i$ is computed.
\colorred{
Although we can show that the additional memory requirements are limited, 
  such analysis is only relevant for general sparse matrices, and not 
  for lattices where no matrix is stored.
Based on Algorithm \ref{HPtask1} we derive later a lattice version 
  that requires no additional memory.
Exploiting the algorithm for general sparse matrices is the topic of 
  future research.
}


Finally, the algorithm computes the permutation $perm$ which 
  orders vertices with the same color in adjacent positions, 
  so that the permuted $A(perm,perm)$ is block diagonal as in
  Figure \ref{fig:ColorPerm}, but has also a similar hierarchical,
  off-diagonal structure.
To compute the quadratures, it is more convenient to permute the rows of the
  vector with the inverse permutation of $perm$, i.e., 
  $iperm(perm) = 1:N$.
Given a vector $z$, $z^T A(perm,perm) z = z(iperm)^T A z(iperm)$.

\subsection{Generating the probing basis}
\label{sec:ProbingSequence}
Assume for the moment that the hierarchy of colors generated by 
  Algorithm \ref{HPtask1} is a power of two, i.e., 
  the nodes are split in two colors at the first level, 
  then the vertices of each color are split into two colors, and so on
  until the base level. 
Obviously, this requires $N=2^m$.
Consider also the permuted colored matrix $A(perm,perm)$ so that colors
  appear in the block diagonal.
In the beginning of Section \ref{section:HP} we saw the Hadamard vectors required for probing
  the first two levels of this recursion for a $8\times 8$ matrix: 
  $[h_0, h_4]$ and $[h_0, h_4, h_2, h_6,]$.
If we denote by $\mathbf{1}_k = [1, \ldots, 1]^T$ the vector of $k$ ones,
  we note that these can be written as:
\begin{eqnarray*}
          [h_0, h_4] &=& H_2\otimes \mathbf{1}_4, \\
{[h_0, h_4, h_2, h_6]} &=& 
 \left[ H_2 \otimes H_2( : ,1),\ H_2 \otimes H_2(:,2)\right] \otimes\mathbf{1}_2. 
\end{eqnarray*}
This pattern extends to any recursion level $k=1,2,\ldots, \log_2 N$. 
If we denote by $Z^{(k)}$ the Hadamard vectors that span the $k$-th level 
  probing vectors, these are obtained by the following recursion:
\begin{eqnarray}
          \tilde Z^{(1)} &=& H_2, \nonumber \\
          \tilde Z^{(k)} &=& \left[ \tilde Z^{(k-1)} \otimes H_2(:,1),\
				    \tilde Z^{(k-1)} \otimes H_2(:,2)\right], \nonumber \\
	  Z^{(k)} & = &\tilde Z^{(k)} \otimes\mathbf{1}_{N/2^{k}}.
						\label{Hperm-recursive}
\end{eqnarray}
Intuitively, this says that at every level, we should repeat the pattern
  devised in the previous level to double the domains 
  for the first $2^{k-1}$ vectors (Kronecker product with $[1, 1]^T$), 
  and then should split each basic subdomain in two opposites 
  (Kronecker product with $[1,-1]^T$).

\colorred{
The hierarchy $Z^{(k-1)} = Z^{(k)}( : ,1:2^{k-1})$ implies that
  quadratures performed with $Z^{(k-1)}$ can be reused if we need 
  to increase the probing level.
}
To obtain the $m$-th probing vector, therefore, we can consider
  the $m$-th vector of $Z^{(\log_2 N)}$.
Its rows can be constructed piece by piece recursively through 
  (\ref{Hperm-recursive}) and without constructing all $Z^{(\log_2 N)}$.
In fact, we can even avoid the recursive construction and compute 
  any arbitrary element of $Z^{(\log_2 N)}(i,m)$ directly. 
This is useful in parallel computing where each processor generates 
  only the local rows of this vector.
The reason is that recursion (\ref{Hperm-recursive}) produces a 
  known permutation of the natural order of the Hadamard matrix, 
  specifically the column indices are:  
\begin{equation}
0, {N/2}, {N/4}, {3N/4}, {N/8}, {5N/8}, {3N/8}, {7N/8}, \ldots.
\label{HpermSequence}
\end{equation}
We can compute a-priori this column permutation array, $Hperm$, for all $N$,
  or for as many vectors as we plan to use in the MC estimator. 
Given also the inverse hierarchical permutation $iperm$ from the previous 
  section, \colorred{
  the $i$-th element of the $m$-th probing vector 
  can be be computed directly through (\ref{eq:HadamardElements}) as:
}
\begin{equation}
	z_m(i) = H_N(iperm(i), Hperm(m)). \label{row-col-permutation}
\end{equation}

\colorred{
We observe now that the assumption that each subgroup is colored with
  exactly two colors is not necessary.
The ordering given in (\ref{HpermSequence}) is the same if each subgroup 
  is colored by any power of two colors, which could be different 
  at different levels.
For example, the nodes mightbe split in four colors at the first level, 
  then the vertices of each color are split into eight colors, and so on.
The sequence (\ref{HpermSequence}) is built on the smallest increment 
  of powers of two and thus subsumes any higher powers.
}

We can extend the above ideas to generate the probing basis \colorred{
  for arbitrary $N$,
} 
  when at every level each color block is split into exactly the same 
  (possibly non-power of two) colors.
For example, at the first level we split the graph into 3 colors, 
  at level two, each of the 3 color blocks is colored with exactly 5 colors, 
  at level three, each of the 5 color blocks is colored with exactly 2 colors, 
  and so on. 
The problem is that Hadamard matrices do not exist for arbitrary 
  dimensions.
For example, for 3 probing vectors, there is no orthogonal basis $Z$ 
  of $\pm1$ elements, such that $ZZ^T = I$.
In this general case, we must resort to the $N$-th roots of unity, 
  i.e., the Fourier matrices $F_n$.

Assume that the number of colors at level $k$ is $c(k)$ for all 
  blocks at that level, then the probing basis is constructed
  recursively as:
\begin{eqnarray}
          \tilde Z^{(1)} &=& F_{C(1)}, \nonumber\\
          \tilde Z^{(k)} &=& \left[ \tilde Z^{(k-1)} \otimes F_{c(k)}( : ,1),\ldots, 
	\tilde Z^{(k-1)} \otimes F_{c(k)}( :,c(k))\right], \nonumber \\
          Z^{(k)} &=& \tilde Z^{(k)} \otimes\mathbf{1}_{N/\gamma_k}.
 					\label{Hperm-general-recursive}\\
	\mbox{ where } \gamma_k &=& \prod_{i=1}^{k} c(i). \nonumber
\end{eqnarray}
By construction, the vectors of $Z^{(k-1)}$ are contained in $Z^{(k)}$,
  and any arbitrary vector can be generated with a simple recursive algorithm.
However, we have introduced complex arithmetic which doubles the computational
  cost for real matrices.
\colorred{
On the other hand, if a $c(k)$ is a power of two, its $F_{c(k)}$ can 
  be replaced by $H_{c(k)}$. 
This can be useful when the non-power of two colors appear only at later
  recursion levels for which the number of probing vectors is large
  and may not be used, or when only one or two $F_{c(k)}$ will suffice.
}

\colorred{
To summarize, we have provided an inexpensive way to generate, for any 
  matrix size, an arbitrary vector of the hierarchical probing sequence through 
  (\ref{Hperm-general-recursive}), as long as 
  the number of colors is the same within the same level for each subgraph.
If, in addition, the matrix size and the color numbers are powers of two, 
  (\ref{HpermSequence}--\ref{row-col-permutation}) provide an even 
  simpler way to generate the probing sequence. 
In LQCD, many of the lattices fall in this last category.
}

We end this section with an open problem that is 
  reserved for future work, \colorred{
  as it is only encountered in general sparse matrices and not in lattices.
}
When blocks at the same level are not colored with the same number of 
  colors, the probing basis cannot be spanned incrementally 
  by Fourier or Hadamard matrices.
Consider for example two colors at level 1. At level 2, the first block is 
  split into 3 colors and the second color split in 2 colors.
Then, for the first level, $Z^{(1)}$ is:
\begin{equation*}
Z^{(1)} =\\
\left[ \begin{array}{cc}
   1 & 0 \\
   1 & 0 \\
   1 & 0 \\
   0 & 1 \\
   0 & 1 \\
\end{array} \right]  = 
\left[ \begin{array}{lr}
	\mathbf{1}_3 & \mathbf{1}_3\\
	\mathbf{1}_2 & -\mathbf{1}_2\\
\end{array} \right],
\end{equation*}
but $Z^{(2)}$ can only be given as $F_5$. 
Moreover, no two columns of $F_5$ span $Z^{(1)}$.
A possible solution for general matrices would be to 
  modify Algorithm \ref{HPtask1} to enforce 
  the same number of colors per block at the same level. 
At the lower levels (which are the farthest distances 
  in the graph) one may choose to stop the recursion and leave 
  the $\mathbf{1}_{N/\gamma_k}$ from (\ref{Hperm-general-recursive}),
  or replace it with a random vector of size $N/\gamma_k$.

\section{Hierarchical probing on lattices}
\colorred{
Uniform $d$-D lattices allow for a far more efficient
  implementation of the hierarchical coloring algorithm,
  based entirely on bit-arithmetic, and guarantee the existence
  of a hierarchical probing basis.
}

Consider first the 1-D case, where the lattice has $N=2^k$ points, where
$k = \log_2 N$, which guarantees the 2-colorability of the 1-D torus.
Any point has a coordinate $0\leq x \leq N-1$ with a binary representation:
   $[b_k, b_{k-1}, \ldots, b_1] =  \mbox{dec2bin}(x)$.
At the first level, the distance-1 coloring is simply red-black
  (we associate red with 0 and black with 1), and
  $x$ gets the color of its least significant bit (LSB), $b_1$. 
In the coloring permutation, we order first the $N/2$ red nodes.
At the second level, we consider red and black points separately and 
split each color again, but now based on the second bit $b_2$.
Thus, points $[**\ldots **00]$ and $[**\ldots **10]$ take different colors,
  and by construction all colors are given hierarchically.
The second level permutation will not mix nodes between the first two halves 
  of the first level, but will permute nodes within the respective halves,
  i.e., points with $0$ in the LSB always appear in the first
  half of the permutation.
The process is repeated recursively for each color, until all points have 
  a different color.

The binary tree built by the recursive algorithm splits the 
  points of a subtree in half at the $m$-th level based on $b_m$.
Thus, to find the final permutation we trace the path from the root to a leaf,
  producing the binary string: $[b_1 b_2 \ldots b_k]$, which is 
  the bit reversed string for $x$.
Denote by $P$ the final permutation array such that node $x=0,\ldots, N-1$ 
  in the original ordering is found in location $P(x)$ of the final 
  permutation.
This is the same inverse permutation $iperm$ of Section \ref{sec:hierPerm},
  only with index numbering starting at 0.
Then, $P(x) = $bin2dec(bitreverse(dec2bin($x$))) and the computation 
  is completely independent for any $i$.

Extending to torus lattices of $d$ dimensions,
  where $N = \prod_{j=1}^d 2^{k_j}$, has three complications:
First, the subgraph of the same color nodes is not a conformal uniform lattice.
Second, the geometry does not allow a simple bit reversal algorithm.
Third, not all dimensions have the same size ($k_j\neq k_i$).
The following sections address these.

\subsection{Splitting color blocks into conformal $d$-D lattices}
Consider a point with $d$ coordinates ($x_1, x_2, \ldots , x_d$).
Let $[b^j_{k_j}, \ldots , b^j_2, b^j_1]$
be the binary representation of coordinate $x_j$ with $0\leq x_j < 2^{k_j}$.
We know that uniform lattices are 2-colorable, so at the first level, 
  red black ordering involves the least significant bit 
  of all coordinates.
The color assigned to the point is mod$(\sum_{j=1}^d b^j_1, 2)$.
However, the red partition, which is half of the lattice points, 
is not a regular $d$-dimensional torus. 
Every red point is distance-2 away from any red neighbor, and therefore 
 it has more neighbors 
 (e.g., in case of 2-D it is connected with 8 neighbors, in 3-D with 18, and 
 so on). 
To facilitate a recursive approach, we observe
 that the reds can be split into $2^{d-1}$ $d$-dimensional sublattices, 
 if we consider them in groups of every other row in each dimension.
Similarly for the blacks.
For the 2-D case this is shown in Figure \ref{fig:2Dsplit}.
\begin{figure}[th]
\centering
\includegraphics[width=0.4\textwidth]{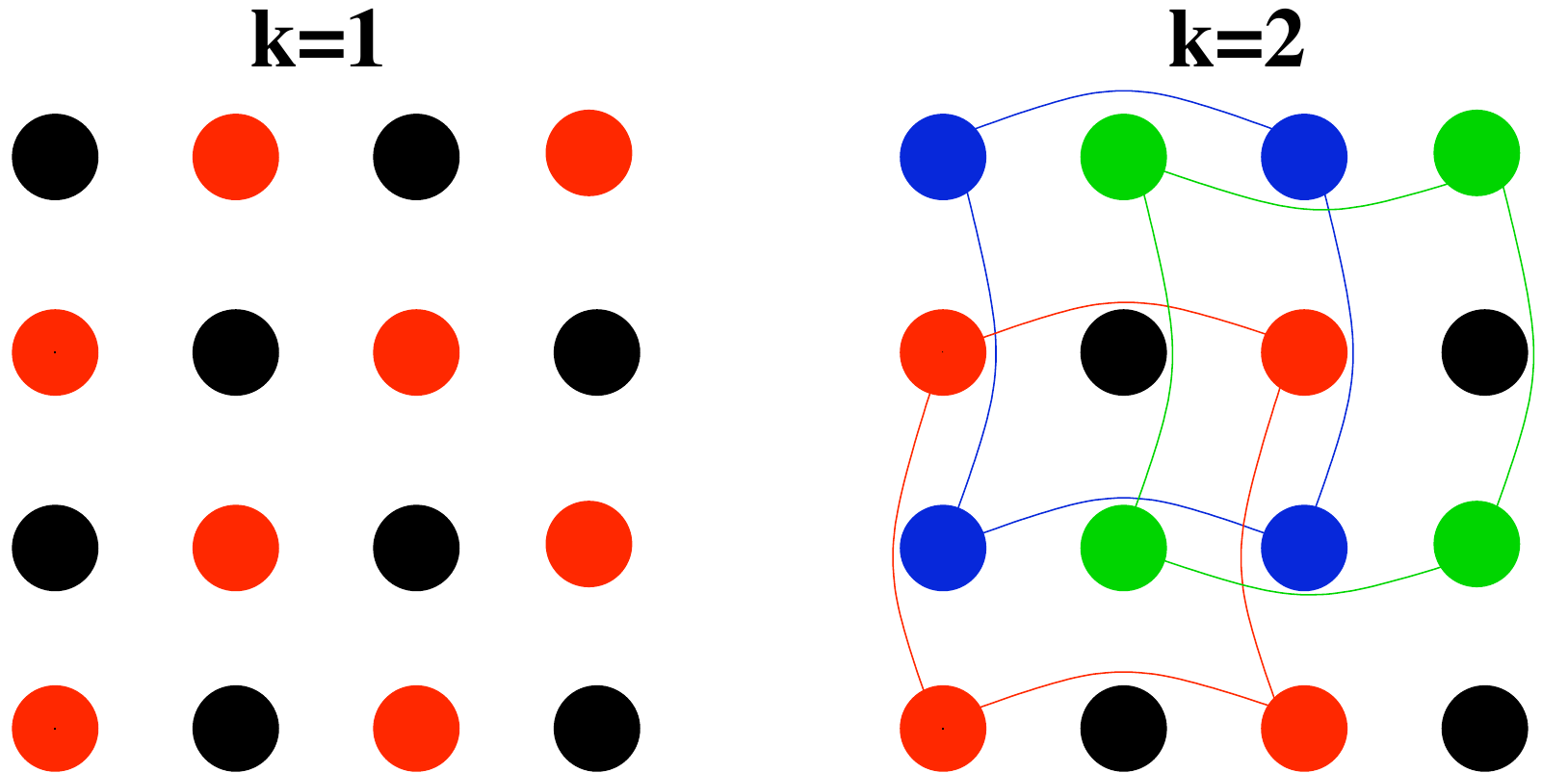}
\caption{When doubling the probing distance (here from 1 to 2) we 
  first split the 2-D grid to four conformal 2-D subgrids. Red nodes 
  split to two $2\times 2$ grids (red and green), and similarly 
  black nodes split to blues and black.
  Smaller 2-D grids can then be red-black ordered.
        }
\label{fig:2Dsplit}
\end{figure}

This partitioning is obtained based on the value of the binary string:
  $[b^1_1, b^2_1, \ldots , b^d_1]$.
For each value, the resulting sublattice contains all points with the given
  least significant bits in its $d$ coordinates.
Because each coordinate loses one bit, the size of each sublattice
  is $\prod_{j=1}^d 2^{k_j-1}$.
At this second level, each of the $2^d$ sublattices can be red-blacked 
  colored independently and with distinct colors for each sublattice.
As long as we remember which were the reds in the first level, 
  the coloring can be hierarchical.

\subsection{Facilitating bit reversal in higher dimensions}
The above splitting based on the LSBs from the $d$ coordinates does not 
order the adjacent colors together. For example, the partitioning at the 
first level of $d=2$ gives four sublattices (00,01,10,11) of which the 
00 and 11 are reds while 01 and 10 are blacks. 
We can recursively continue partitioning and coloring the sublattices.
However, if we concatenate at every level the new 2 bits from the 2 
  coordinates, as in the bit reversed pattern in the 1-D case,
  the resulting ordering is not hierarchical.
In our example, all red points in the first level are ordered in the first
  half, but at the second level, the colors associated with the 00 reds 
  will be in the first quarter of the ordering, while the colors associated
  with the 11 reds will be in the fourth quarter of the ordering.
Since the hierarchical ordering is critical for reusing previous work, 
  we order the four sublattices not in the natural order (00,01,10,11)
  but in a red black order: (00 11 01 10). 
Algorithm \ref{alg:slowRB} produces this Red-Black reordering in $d$ dimensions.

A more computationally convenient way to obtain the $RB$ permutation 
  is based on the fact that every point on the stencil has neighbors 
  of opposite color. In other words, $color( [x_1,\ldots ,x_d] ) = 
 \lnot color([x_1,\ldots, x_d]\pm e_j)$,
  where $e_j$ is the unit row-vector in the $j$ dimension, $j=1,\ldots, d$,
  and $\lnot$ is the logical not.
With two points per dimension, in one dimension the colors are $c_1 = [0, 1]$.
Inductively, if the colors in dimension $d-1$ are $c_{d-1}$, the second $d-1$ 
  plane in dimension $d$ will have the opposite colors, and thus:
  $c_d = [ c_{d-1}, \lnot c_{d-1}]$.
Therefore, we can create the $RB$ with only a check per point instead
  of counting coordinate bits.
This is shown in Algorithm \ref{alg-rb1}.

\noindent
\begin{minipage}[t]{0.47\textwidth}
\begin{algorithm}[H]
\caption{\label{alg:slowRB} Red-Black order of the $2^d$ torus (slow)}
\begin{tabbing}xxx\=xxx\=xxx\=\kill
$RB$ = bitarray($2^d,d$)\\
$reds = 0,\ blacks = 2^{d-1}$\\
{\bf for} $i=0:2^d-1$\\
\>   {\bf if} dec2bin($i$) has even number of bits \\
\>\>      $newbits$ = dec2bin($reds, d$) \\
\>\>      $reds = reds+1$ \\
\>   {\bf else} \\
\>\>      $newbits$ = dec2bin($blacks, d$) \\
\>\>      $blacks = blacks+1 $\\
\>   $RB(i,:) = newbits$
\end{tabbing}
\end{algorithm}
\end{minipage} \ \ \
\begin{minipage}[t]{0.47\textwidth}
\begin{algorithm}[H]
\caption{ \label{alg-rb1} Red-Black order of the $2^d$ torus (fast)}
\begin{tabbing}xxx\=xxx\=xxx\=\kill 
$c_0 = 0$\\
{\bf for} $j=1:d$\\
\>   $c_j = [c_{j-1}, \lnot c_{j-1}]$ \\
$RB$ = bitarray($2^d,d$)\\
$reds = 0,\ blacks = 2^{d-1}$\\
{\bf for} $i=0:2^d-1$\\
\>   {\bf if} $c_d(i) == 0$  \\
\>\>      $newbits$ = dec2bin($reds, d$) \\
\>\>      $reds = reds+1$ \\
\>   {\bf else} \\
\>\>      $newbits$ = dec2bin($blacks, d$) \\
\>\>      $blacks = blacks+1 $\\
\>   $RB(i,:) = newbits$
\end{tabbing}
\end{algorithm}
\
\end{minipage}

We are now ready to combine the Red-Black reordering 
  with the bit-reversing scheme to address the $d$-dimensional case.
First, assume that the lattice has the same size in each dimension, 
i.e., $k_j = k, \forall j=1,\ldots , d$. Then the needed permutation 
is given by Algorithm \ref{alg-perm0}.

\begin{algorithm}
\caption{ \label{alg-perm0} 
Hierarchical permutation of the lattice -- case $k_j=k$}
\begin{tabbing}xxx\=xxx\=xxxi\=\kill 
\% {\bf Input:}\\
\% the coordinates of a point $x = (x_1, x_2, \ldots , x_d)$\\
\% the global $RB$ array produced by Algorithm \ref{alg-rb1}\\
\% {\bf Output:}\\
\% The location in which $x$ is found in the hierarchical permutation\\
{\bf function} loc = {\sc LatticeHierPermutation0}($(x_1, x_2, \ldots , x_d)$)\\
\>  \% Make a $d\times k$ table of all the coordinate bits\\
\>  {\bf for} $j=1:d$\\
\>\>    $(b^j_k, \ldots , b^j_2, b^j_1) = \mbox{dec2bin}(x_j)$\\
\>  \% Accumulate bit-reversed order. Start from LSB\\
\>  $loc = [\ ]$ \\
\>  {\bf for} $m=1:k$\\
\>\>   \% A vertical section of bits. Take the m-th bit of all coordinates\\
\>\>   \% and permute it to the corresponding red-black order \\
\>\>   $(s^1,\ldots ,s^d) = 
	RB(\mbox{bin2dec}(b^1_m, b^2_m, \ldots , b^d_m))$\\
\>\>   \% Append this string to create the reverse order string\\
\>\>	$loc = [loc,~~(s^1,\ldots ,s^d) ]$\\
\>  return bin2dec($loc$)
\end{tabbing}
\end{algorithm}

\subsection{Lattices with different sizes per dimension}
At every recursive level, our algorithm splits the size of each dimension 
  in half (removing one bit), until there is only 1 node per dimension.
When the dimensions do not all have the same size, 
  some of the dimensions reach 1 node first and beyond that point 
  they are not subdivided.
If $m$ from the $d$ dimensions have reached size 1, 
  the above algorithm should continue as in a $d-m$ dimension lattice, 
  at every level concatenating only the active $d-m$ bits in $loc$.
In this case, however, the red-black permutation $RB$ should 
  correspond to that of a $d-m$ dimensional lattice.
The following results allow us to avoid computing and storing 
  $RB_j$ for each $j=1,\ldots ,d$. 
As before, we consider $c_d$ the array of 0/1 colors of the 
  2-point, $d$-dimensional torus.
\begin{lemma}
\label{lemma-doubledim}
For any $d>0$, $c_d(2i) = c_{d-1}(i), \forall i=0,\ldots , 2^{d-1}-1$.
\end{lemma}
\begin{proof}
We use induction on $d$. For $d=2$, $c_2 = [0, 1, 1, 0]$ and the result holds.
Assume the result holds for any dimension $d-1$ or lower. 
Then for $d$ dimensions, since the first half of $c_d$ is the same as 
  $c_{d-1}$, for $i=0,\ldots 2^{d-1}-1$, we have
$$
\begin{array}{cclr}
c_d(2i) & = & \lnot c_{d}(2i-2^{d-1}) = 
\lnot c_{d-1}(2i-2^{d-1})            &\mbox{(recursive definition of $c_d$)}\\
	& = &\lnot c_{d-2}(i-2^{d-2}) = 
		\lnot c_{d}(i-2^{d-2})  &\mbox{(inductive hypothesis)}\\
	& = &\lnot c_{d-2}(i-2^{d-2}) = \lnot (\lnot c_{d-1}(i))
	= c_{d-1}(i) 	             &\mbox{(recursive definition of $c_d$)}.
\end{array}
$$

\end{proof}
%
\begin{lemma}
\label{lemma-diffpairs}
For any $d>0$, $c_d(2i) = \lnot c_d(2i+1), \ i=0,1,\ldots ,2^{d-1}-1$.
\end{lemma}
\begin{proof}
Because $c_d$ are the colors of the two-point, $d$ dimensional torus, 
  every even point $2i$ is the beginning of a new 1-D line and thus has 
  a different color from its neighbor $2i+1$. 
It can also be proved inductively, since by construction $2i$ and $2i+1$ 
  cannot be split across $c_{d-1}$ and $c_d$.
\end{proof}
\begin{lemma}
\label{lemma-valsRB}
For any $d>0$ the values of $RB_d(i),\ i=0,\ldots ,2^d-1$ are given by:
$$ RB_d(i) = \left\{ \begin{array}{ll}
		\lfloor i/2 \rfloor, & \mbox{ if } c_d(i) = 0 \\
		\lfloor i/2 \rfloor + 2^{d-1},& \mbox{ if } c_d(i) = 1
	   	    \end{array} \right..$$
\end{lemma}
\begin{proof}
Because of Lemma~\ref{lemma-diffpairs}, after every pair of indices
  ($2i, 2i+1$) is considered, the number of reds or blacks increases
  only by 1.
Algorithm \ref{alg-perm0} sends all red ($c_d(i)=0$) points $i$ to the 
  first half of the permutation in the order they are considered, which 
  increases by 1 every two steps. Hence the first part of the equation.
Black colors are sent to the second half, which completes the proof.
\end{proof}

We can now show how $RB_m, m<d,$ can be obtained from $RB_d$.
\begin{theorem}
Let $RB_d$ be the permutation array that groups together the same colors 
  in a red-black ordering of the two-point, $d$ dimensional lattice, 
  as produced by Algorithm \ref{alg-rb1}. For any $0<m<d$, 
  $$RB_m(i) = \lfloor RB_d(i 2^{d-m}) / 2^{d-m} \rfloor,
	\ i=0,\ldots , 2^{m}-1.$$
\end{theorem}
\begin{proof}
We show first for $m=d-1$.
Because of Lemma~\ref{lemma-doubledim}, we consider the even points
  in $RB_d$. 
Assume first $c_d(2i) = c_{d-1}(i) = 0$. From Lemma~\ref{lemma-valsRB} we have,
$RB_d(2i) = \lfloor 2i/2 \rfloor = i$.
Then, $RB_{d-1}(i) = \lfloor i/2 \rfloor = \lfloor RB_d(2i)/2 \rfloor$.
Now assume $c_d(2i) = c_{d-1}(i) = 1$. From Lemma~\ref{lemma-valsRB} we have,
$RB_d(2i) = 2^{d-1} + \lfloor 2i/2 \rfloor = 2^{d-1} + i$, 
and therefore
$RB_{d-1}(i) = 2^{d-2}+\lfloor i/2 \rfloor = 2^{d-2} + 
	\lfloor (RB_d(2i)-2^{d-1})/2 \rfloor  = 
	\lfloor RB_d(2i)/2 \rfloor$, which
proves the formula for both colors.
A simple inductive argument proves the result for any $m=1,\ldots,d-2$.
\end{proof}

The theorem says that given $RB_d$ in bit format, $RB_m$ is obtained as 
  the left (most significant) $m$ bits of every $2^{d-m}$ number in $RB_d$.
We now have all the pieces needed to modify Algorithm \ref{alg-perm0}
  to produce the permutation of the hierarchical coloring of $d$ dimensional 
  lattice torus of size $N=\prod 2^{k_j}, \forall j=1,\ldots , d$.

\begin{algorithm}
\caption{\label{alg-perm}
Hierarchical permutation of the lattice -- case $2^{k_i} \neq 2^{k_j}$}
\begin{tabbing}xxx\=xxx\=xxx\=xxx\=\kill
\% {\bf Input:}\\
\% the coordinates of a point $x = (x_1, x_2, \ldots , x_d)$\\
\% the global $RB$ array produced by Algorithm \ref{alg-rb1}\\
\% {\bf Output:}\\
\% The location in which $x$ is found in the hierarchical permutation \\
{\bf function} loc = {\sc LatticeHierPermutation}($(x_1, x_2, \ldots , x_d)$)\\
\>  \% Make a $d\times \max (k_j)$ table of all the coordinate bits\\
\>  \% Dimensions with smaller sizes only have up to $k_j$ bits set\\
\>  {\bf for} $j=1:d$\\
\>\>    $(b^j_{k_j}, \ldots , b^j_2, b^j_1) = \mbox{dec2bin}(x_j)$\\
\>  \% Accumulate bit-reversed order. Start from LSB\\
\>  $loc =[\ ]$ \\
\>  {\bf for} $m=1:\max(k_j)$\\
\>\>   \% A vertical section of bits. Take the m-th bit of all coordinates\\
\>\>   \% in dimensions that can still be subdivided ($m\leq k_j$).\\
\>\>   \% Record number of such dimensions\\
\>\>   $activeDims = 0$\\
\>\>   $bits = [\ ]$\\
\>\>   {\bf for} $j=1:d$\\
\>\>\>    {\bf if} ($m \leq k_j$) \\
\>\>\>\>    $bits = [bits, ~~ b^j_m]$\\
\>\>\>\>    $activeDims = activeDims+1$\\
\>\>   \% permute it to the corresponding red-black order using $RB_m$\\
\>\>   $index = \mbox{bin2dec}(bits) 2^{d-activeDims}$\\
\>\>   $(s^1,\ldots ,s^{activeDims}) = \lfloor RB(index)/2^{d-activeDims}\rfloor$ \\
\>\>   \% Append this string to create the reverse order string\\
\>\>	$loc = [loc,~~(s^1,\ldots ,s^{activeDims}) ]$\\
\>  {\bf return} bin2dec($loc$)
\end{tabbing}
\end{algorithm}

For $d>1$, Algorithm \ref{alg-perm} is not equivalent to Algorithm~\ref{HPtask1}
  \colorred{
   because it pre-splits color subgroups into conformal lattices.
In general, the difference in the number of colors is small.
At level $m=0,1,\ldots$, Algorithm \ref{alg-perm} performs (at least) a 
  distance-$2^m$ coloring and produces $2^{dm+1}$ colors.
}

  \colorred{
For classic probing, the minimum number of colors required for
  distance-$2^m$ coloring of lattices is not known for $d>2$ 
  \cite{Blaum98interleavingschemes}.
An obvious lower bound is the number of lattice sites in the 
  ``unit sphere'' of graph diameter $2^{m}$.
If $\left( \begin{array}{c} d \\ i \end{array} \right)$ denotes the 
  binomial coefficient, with a 0 value if $d<i$, 
  the lower bound is given by
  \cite[\mbox{Theorem 2.7}]{ComputingContinuousDiscretely}:
$$\sum_{i=0}^d 
	\left( \begin{array}{c} d \\ i \end{array} \right) 
	\left( \begin{array}{c} d-i+2^{m-1} \\ d \end{array} \right).
$$
For sufficiently large distances, this is $O( 2^{3m-1}/3)$ for $d=3$,
  and $O(2^{4m-3}/3)$ for $d=4$.
Thus, we can bound asymptotically how many more colors our method gives:
$$\frac{\mbox{Number of colors in hierarchical probing}}
       {\mbox{Number of colors in classic probing}}  \left\{ 
	\begin{tabular}{ll}
	$<$ 12, & \mbox{if} d=3 \\
	$<$ 48, & \mbox{if} d=4.
	\end{tabular}
	\right.
$$
}
  \colorred{
In practice, we have observed ratios of 2--3.
On the other hand, because hierarchical probing uses more vectors,
  the variance reduction it achieves 
  when a certain distance coloring completes, i.e., after 
  $2^{dm+1}$ quadratures in the MC estimator, is typically better 
  than classic probing for the same distance.
}

In terms of computational cost, the algorithm is not only tractable
  (compared to classic probing), but it is very efficient.
As an example, producing the hierarchical permutation of a $64^4$ lattice 
  takes about 6 seconds on a Macbook Pro with 2.8 GHz Intel Core 2 Duo.
More importantly, the permutation of each coordinate is obtained 
  independently which facilitates parallel computing.

\cred
\subsection{Coloring lattices with non-power of two sizes}
\label{sec:nopower2}
Consider a lattice of size $N = \prod_{i=1}^d n_i$.
Sometimes, LQCD may generate lattices where one or more $n_i$ are not
  powers of two.
In this case, it is typical that $n_i = 2^mp$, where $p\neq 2$ is a 
  small prime number.
Our hierarchical coloring method works up to $m$ levels,
  but then the remaining subgrids are of odd size in the $i$-th dimension, 
  causing coloring conflicts because of wrap-around.
We show that such a lattice is three-colorable.

\begin{theorem}
A toroidal, uniform lattice of size $N = \prod_{i=1}^d n_i$, where
  one or more $n_i$ are odd, admits a three-coloring with 
  point $x = (x_1, \ldots , x_d)$ receiving color:
\[
  C(x)=\left(\sum_{i=1}^{d} x_i + \sum_{i=1}^{d} \delta(x_i)\right)\mbox{mod} \;3,
\ \mbox{where} \
\delta(x_i) = \left\{ \begin{array}{ll}
	1, & \mbox{if } (x_i = n_i-1) \mbox{ and}\\
	   & \ (n_i-1\;\mbox{mod}\;3 = 0)\\
	0, & \mbox{everywhere else}.
	\end{array}
\right.
\]
\end{theorem}
\begin{proof}
We show that $C(x) \neq C(x')$ for any two points, $x$, $x'$ 
  with $||x-x'||_1 = 1$.
These two points differ by one coordinate, $j$, 
since otherwise they are no longer unit length apart. So,
$C(x)-C(x') =
 \left(\sum_{i=1}^{N} (x_i - x'_i)
  + \sum_{i=1}^{N} (\delta(x_i) - \delta(x'_i))\right) mod \;3 
  =  \left( x_j-x'_j+\delta(x_j)-\delta(x'_j)\right) \; mod \; 3$.
We consider the following cases.

If neither $x$ and $x'$ lie on the boundary of the $j$-th dimension, 
  $x_j \neq n_j -1$, then $\delta(x_j)=\delta(x'_j)=0$, and 
  $C(x)-C(x') = (x_j-x'_j) \; mod \; 3 = \pm 1 \; mod \; 3 \neq 0$. 

Since $x_j,x'_j$  both vary along the j-th dimension, only one of 
  these points can lie on the boundary point of that dimension, 
  consequently, only one of the two deltas can be equal to one.  
Without loss of generality, we assume that $x_j$ is on the 
  boundary of the $j$-th dimension, 
  so $C(x)-C(x')=(x_j-x'_j+\delta(x_j))\; mod \; 3$.
In this case $x_j-x'_j=1$, or in the warp around case, 
  where $x'_j=0,x_j-x'_j=n_j-1$. 
There are two subcases:
\begin{enumerate}
\item $\delta(x_j)=0$, then $x_j=n_j-1$ with $n_j-1 \;mod\;3\neq 0$, 
  so $C(x)-C(x')=1$, or $C(x)-C(x')=(n_j-1 \; mod \; 3)$ and thus is non-zero.

\item $\delta(x_j)=1$, then $x_j = n_j - 1$ and $n_j-1\; mod\; 3 = 0$, 
  so $C(x)-C(x')$ is equal to 
  $(1 + \delta(x_j))\; mod \;3=(1+1)\; mod \;3 \neq 0$,
  or $C(x)-C(x')$ is equal to 
  $(n_j -1 +\delta(x_j)) = (0 + \delta(x_j))\; mod \; 3 =  1 \; mod \; 3 \neq 0$.
\end{enumerate}
\end{proof}


\colorred{
After the method produces the three-coloring, no further hierarchical 
colorings can be produced.
This is because the three-coloring yields blocks of nodes that are not conformal lattices, and are of irregular sizes and shapes. This prevents the use of a method similar to the one described in section 3. Because of this, for lattices which have dimensions with factors other than two, we can proceed only one level further after the factors of two have been exhausted by the hierarchical coloring algorithm. 
This is not a shortcoming in LQCD, since, by construction,
  lattices have dimensions with only one odd factor.
Finally, note that the number of hierarchical probing vectors produced 
  before exhausting the powers of two in each dimension is typically very 
  large already, obviating the use of the last, three-coloring step to 
  produce even more vectors.
}

\cblack

\section{Removing the deterministic bias}
The probing vectors produced in Section \ref{sec:ProbingSequence}
  are deterministic and, even though they give better approximations than 
  random vectors, they introduce a bias.
To avoid this, we can view formula (\ref{Hperm-recursive}) not as a 
  sequence of vectors but {\em as a process of generating an orthogonal 
  basis starting from any vector and following a particular pattern}. 
Therefore, consider a random vector $z_0 \in \mathbb{Z}_2^N$, 
  and $[z_1, z_2, \ldots ,z_m]$ the sequence of vectors produced 
  by (\ref{Hperm-recursive}).
If $\odot$ is the element-wise Hadamard product, 
  the vectors built as 
\begin{equation}
  V = [ z_0\odot z_1, z_0\odot z_2, \ldots, z_0\odot z_m]
  \label{eq:removebias}
\end{equation}
  have the same properties as $Z$, i.e.,  
  $V^TV = Z^TZ$ and $VV^T$ has same non-zero pattern as $ZZ^T$ 
  ($VV^T = (z_0z_0^T)\odot ZZ^T$), but it does not have the bias.

\section{Numerical experiments}
We present a set of numerical examples on control test problems 
  and on a large QCD calculation,
  in order to show the effectiveness of
  hierarchical probing over classic probing, and over standard
  noise Monte Carlo estimators for $\TrAi$.
We also study the effect of removing the bias on convergence.

Our standard control problem is the discretization of the 
  Laplacian on a uniform lattice with periodic boundary conditions.
We control the dimensions (3-D or 4-D), the size per dimension, 
  and the conditioning (and thus the decay of the elements of the inverse)
  by applying a shift to the matrix.
Most importantly, for these matrices we know the trace of the inverse 
  analytically.
We will refer to such problems as Laplacian, with their size implying 
  their dimensionality.

\subsection{Comparison with classic probing}
For this set of experiments we consider a $64^3$ Laplacian, shifted 
  so that its condition number is $O(100)$.
Therefore, its $A^{-1}$ exhibits dominant features on and close to 
  (in a graph theoretical sense) the 
  non-zero structure of $A$, with decay away from it.
The decay rate depends on the conditioning of $A$.
Our methods should be able to pick this structure effectively.

Figure \ref{fig:ClassicProbingComparisons} shows the performance of 
  classic probing, which is a natural benchmark for our methods.
The left graph shows that for larger distance colorings, 
  probing performs extremely well.
For example, with 317 probing vectors, which correspond
  to a 8-distance coloring, we achieve more than two orders reduction in 
  the error. 
Of course, if the approximation is not good enough, this work must be 
  discarded, and the algorithm must be repeated for higher distances.
Hadamard vectors, used in their natural order, do not capture well 
  the nonzero structure of this $A$.

The right graph in Figure \ref{fig:ClassicProbingComparisons}
  shows one way to improve accuracy beyond a certain probing distance.
After using $[0,\ldots, 0, \mathbf{1}_{c(k)}^T, 0, \ldots 0]^T$ as the 
  probing vector for color $k$,
  we continue building a Hadamard matrix in its natural order only for the 
  $c(k)$ coordinates of that color.
If probing has captured the most important parts of the matrix, 
  the remaining parts could be sufficiently approximated by natural 
  order Hadamard vectors.
This is confirmed by the results in the graph, if one knows what initial
  probing distance to pick.
On the other hand, hierarchical probing, which considers all possible 
  levels, achieves better performance than all other combinations.

\begin{figure}[ht]
\includegraphics[width=0.5\textwidth,angle=0]{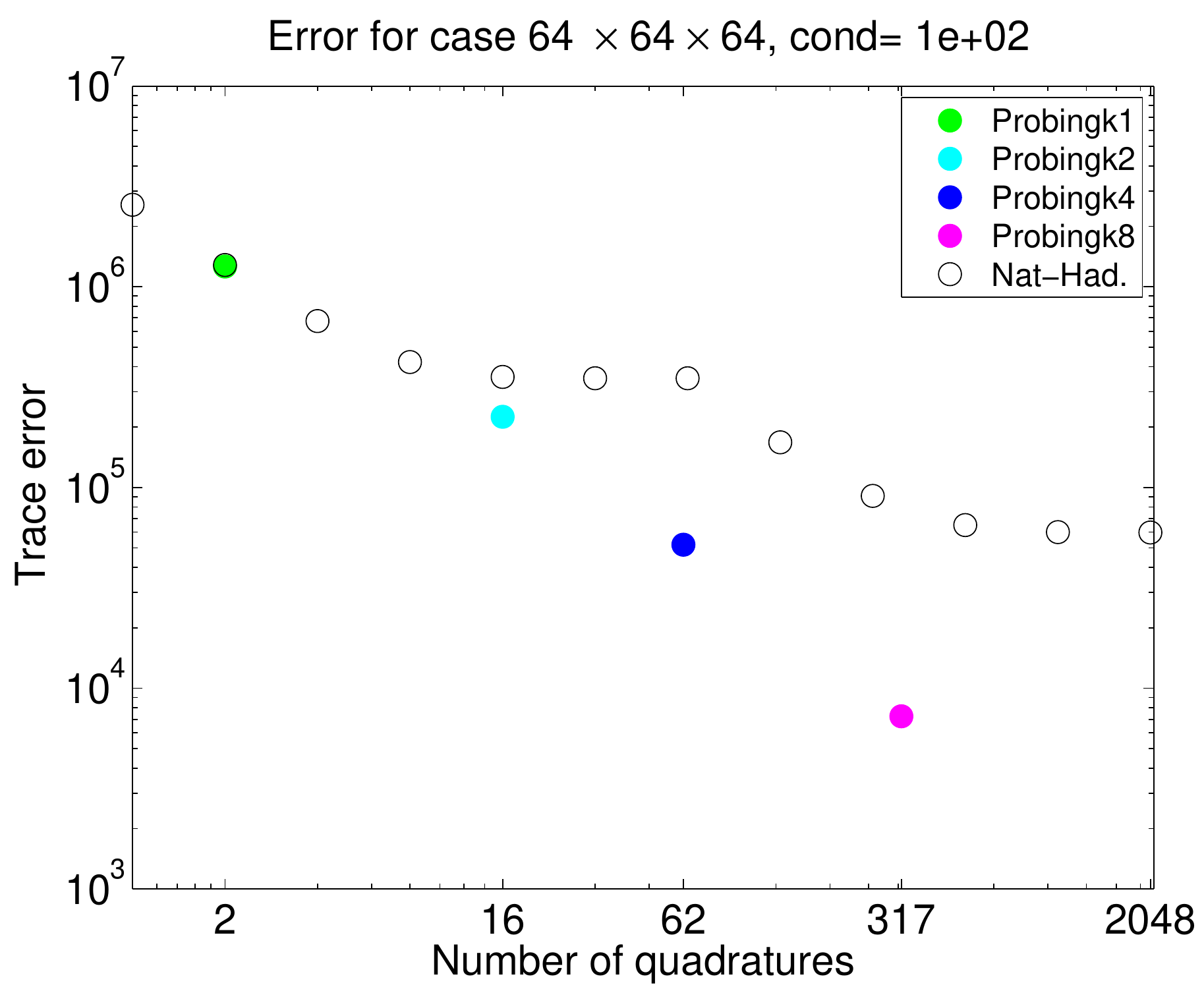}
\includegraphics[width=0.5\textwidth,angle=0]{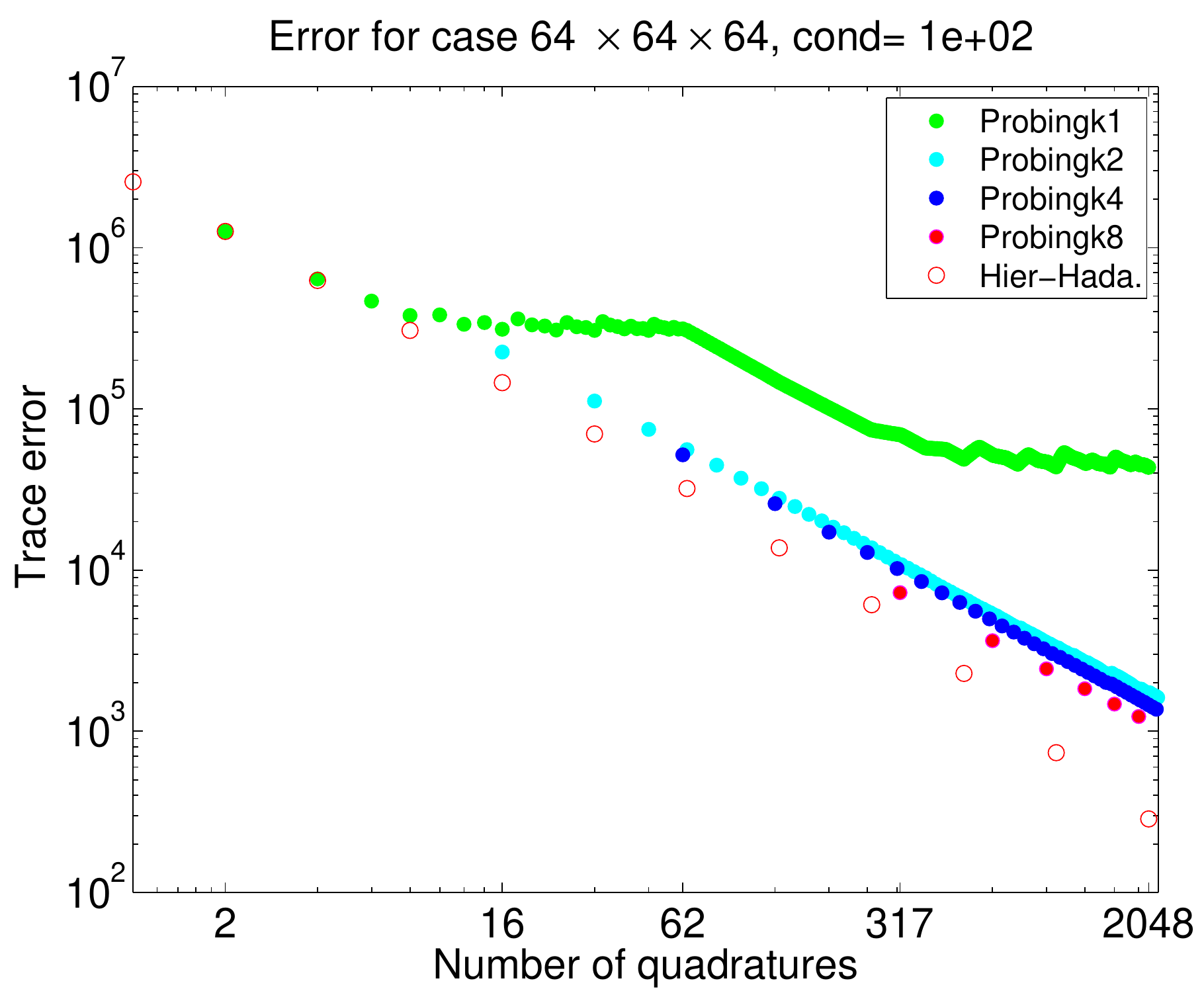}
\caption{ \label{fig:ClassicProbingComparisons}
Error in the $\TrAi$ approximation using the MC method with various
  deterministic vectors.
Classic probing requires 2,16,62, and 317 colors for probing 
  distances 1,2,4, and 8, respectively. 
Left: Classic probing approximates the trace better than the same 
  number of Hadamard vectors taken in their natural order.
  Going to higher distance-$k$ requires discarding previous work. 
Right: Perform distance-$k$ probing, then apply Hadamard in 
  natural order within each color. 
  Performs well, but hierarchical performs even better.
        }
\end{figure}
\begin{figure}[hbt]
\includegraphics[width=0.49\textwidth,angle=0]{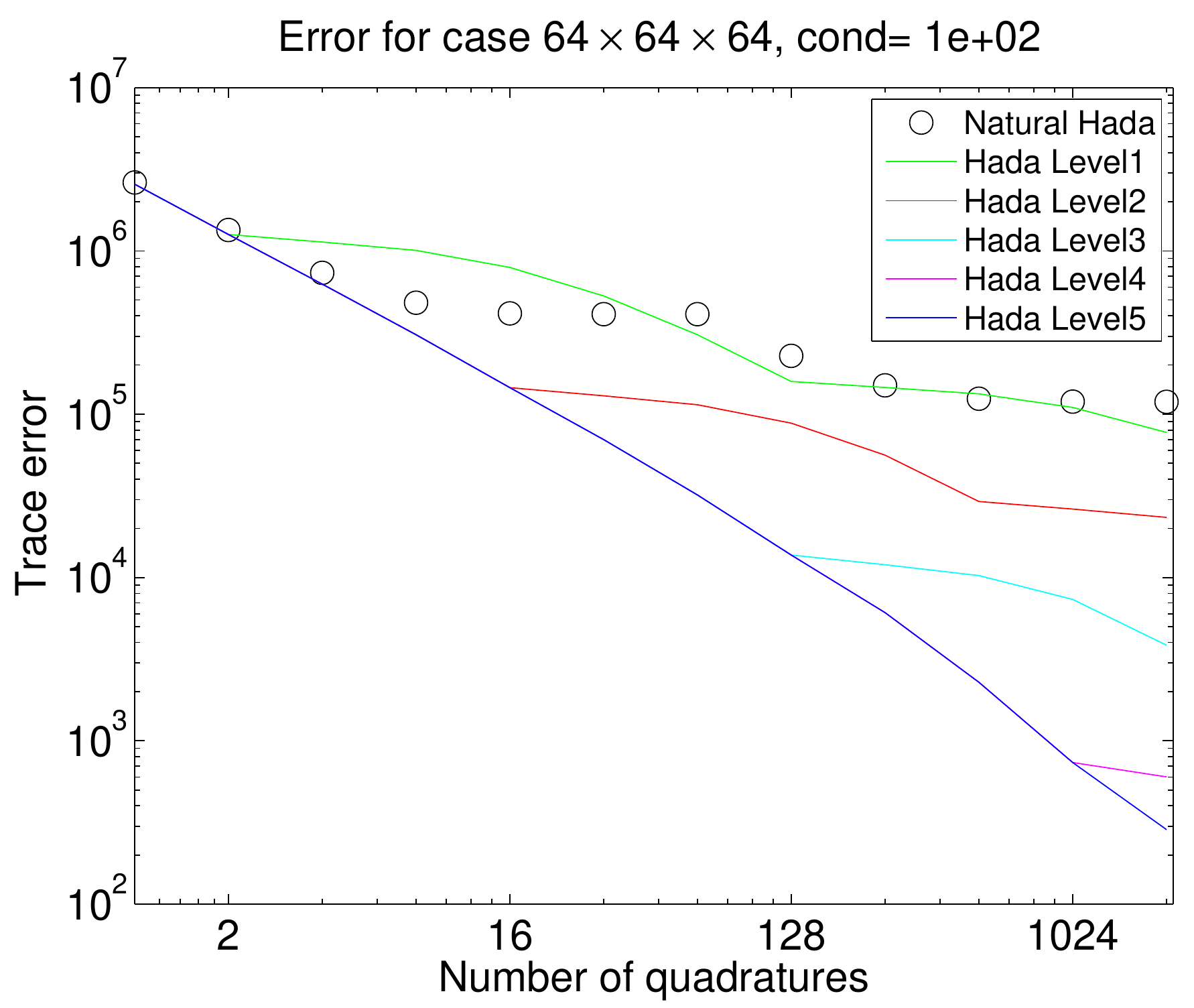}
\includegraphics[width=0.51\textwidth,angle=0]{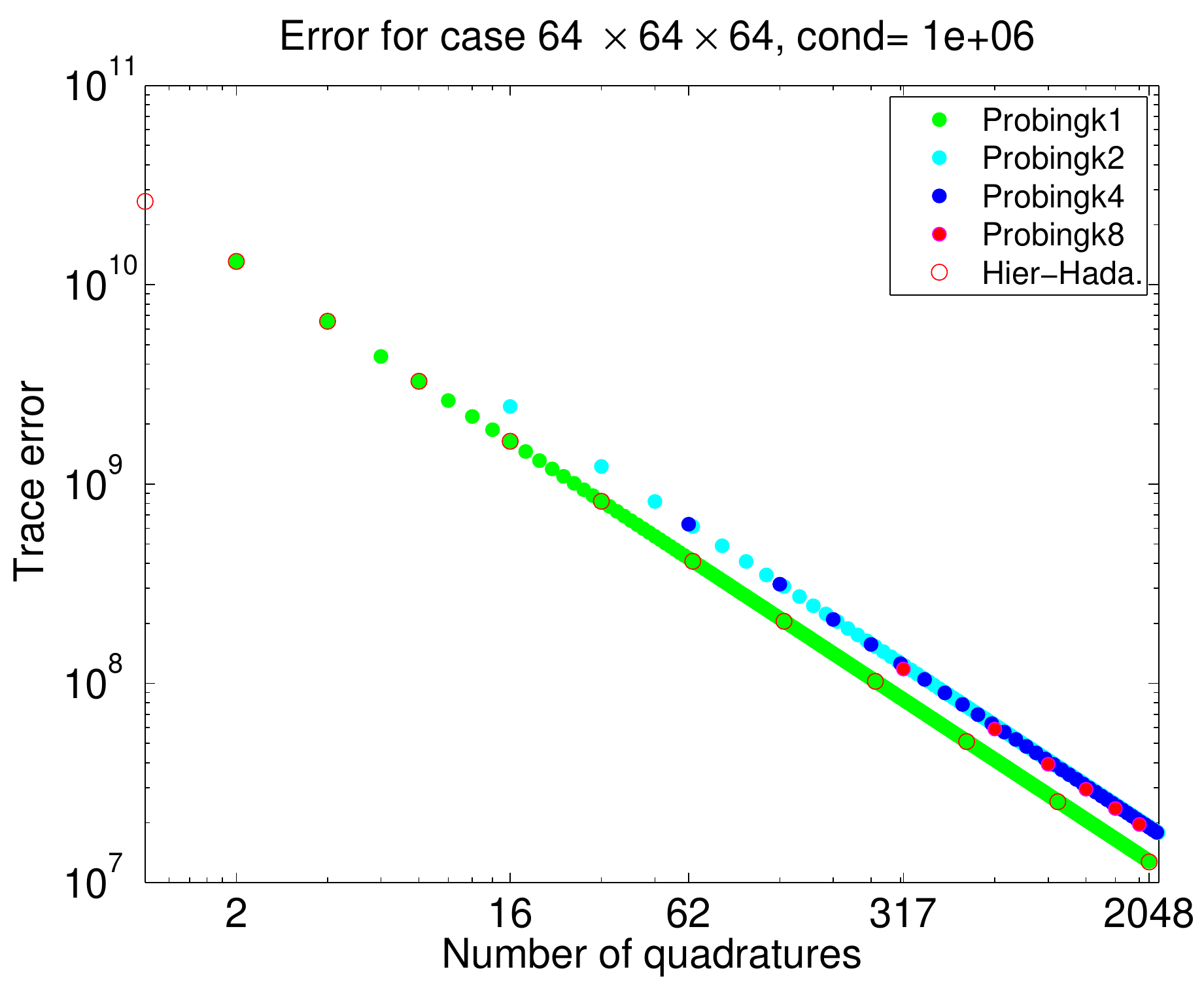}
\caption{ \label{fig:VariousLevels-highCond}
Left: The hierarchical coloring algorithm is stopped after $1,2,3,4,5$ levels
  corresponding to distances $2,4,8,16,32$. The ticks on the x-axis show
  the number of colors for each distance. 
Trace estimation is effective up to the stopped level; beyond that the vectors
  do not capture the remaining areas of large elements in $A^{-1}$.
Compare the results with classic probing in 
  Figure~\ref{fig:ClassicProbingComparisons}, which requires only a few 
  less colors for the same distance.
Right: When the matrix is shifted to have high condition number, the lack
  of structure in $A^{-1}$ causes all methods to produce similar results.
        }
\end{figure}

In Figure \ref{fig:VariousLevels-highCond}, left graph, 
  we stop our recursive algorithm at various levels and use the 
  resulting permutation to generate the vectors for the trace computation.
It is clearly beneficial to allow the recursion to run for all levels.
We also point out that stopping at intermediate levels behaves similarly to
  classic probing with the corresponding distance.
On the right graph of Figure \ref{fig:VariousLevels-highCond},
  we observe no difference between methods for high conditioned matrices.
The reason is that the eigenvector of the smallest eigenvalue of $A$ 
  is the vector of all ones, $\mathbf{1}_N$.
The more ill conditioned $A$ is, the more $A^{-1}$ is dominated by 
  $\mathbf{1}_N\mathbf{1}_N^T$, which has absolutely no variation or pattern.

We point out that the experiments in this subsection did
  not use the bias removing technique that takes a Hadamard product of 
  all vectors in the sequence with the same random $\mathbb{Z}_2$ vector.
This has a severe effect for the Laplacian matrix because the 
  first vector of our Hadamard sequences is $h_0 = \mathbf{1}_N$, 
  the lowest eigenvector.
Even for a well conditioned Laplacian, starting with $h_0$ guarantees 
  that the first trace estimate will have no contribution from other 
  eigenvectors, and thus will have a large error.
 From a statistical point of view, $h_0$ is the worst starting vector 
  for Laplacians, but it better exposes the rate at which error reduces 
  by various methods.

\subsection{Comparison with random-noise Monte Carlo}
Having established that hierarchical probing discovers matrix structure 
  as well as classic probing, we turn to gauge its improvements over
  the standard $\mathbb{Z}_2$ noise MC estimator.
First, we show three sets of graphs for increasing condition numbers of  
  the Laplacian. 
We use the $64^3$, $32^4$, and $64\times 128^2$ lattices, and plot
  the convergence of the trace estimates for hierarchical probing, natural 
  order Hadamard, and for the standard $\mathbb{Z}_2$ random estimator.
Both Hadamard sequences employ the bias removing 
  technique (\ref{eq:removebias}).
As it is typical, the random estimator includes error bars designating
  the two standard deviation confidence
  intervals, $\pm2(\overline{Var}/s)^{1/2}$, 
  where $\overline{Var}$ is the variance estimator.

Figure \ref{fig:vsZ2smallCond} shows the convergence history of the three
  estimators for well conditioned shifted Laplacians, which therefore 
  have prominent structure in $A^{-1}$.
Hierarchical probing exploits this structure, and thus performs much 
  better than the other methods.
Note that the problem on the left graph is identical to the one used
  in the previous section. 
The far better performance of the Hadamard sequences in this case 
  is due to avoiding the eigenvector $h_0$ as the starting vector.
Once again, Hadamard vectors in natural order should only be used
  for special banded matrices.

Figures \ref{fig:vsZ2midCond} and \ref{fig:vsZ2largeCond} show 
  results as the condition number of the problems increase. 
As expected, the advantage of hierarchical probing wanes as the 
  structure of $A^{-1}$ disappears, but there is still no reason
  not to use it as the method still provides improvement, albeit diminishing.
We have included 4-D lattices in our experiments, first because 
  of their use in LQCD, and second because they are more difficult 
  to exploit their structure than lower dimensionality lattices.
For 1-D or 2-D lattices which we do not show, hierarchical probing was
  significantly more efficient.

\begin{figure}[ht]
\includegraphics[width=0.5\textwidth,angle=0]{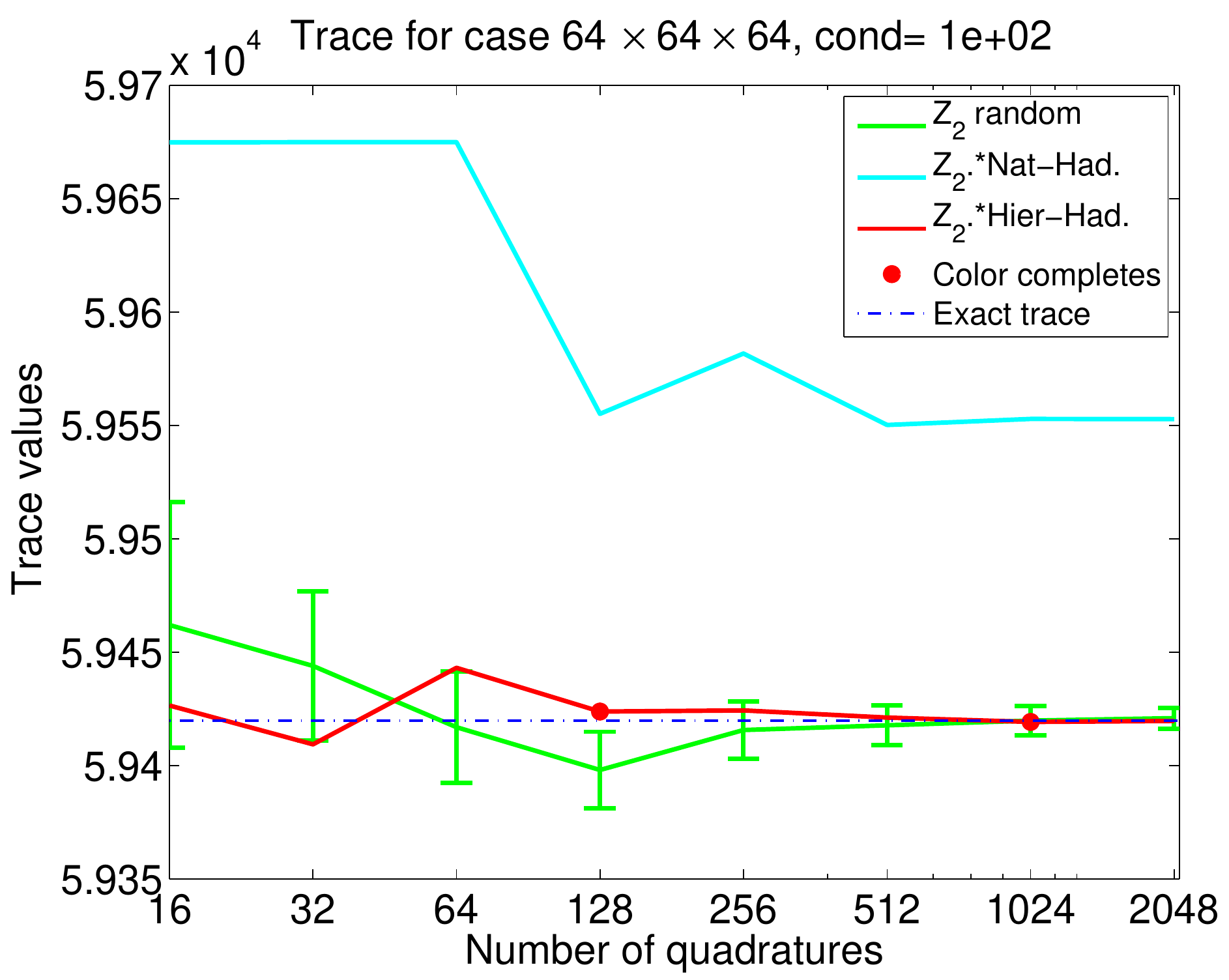}
\includegraphics[width=0.5\textwidth,angle=0]{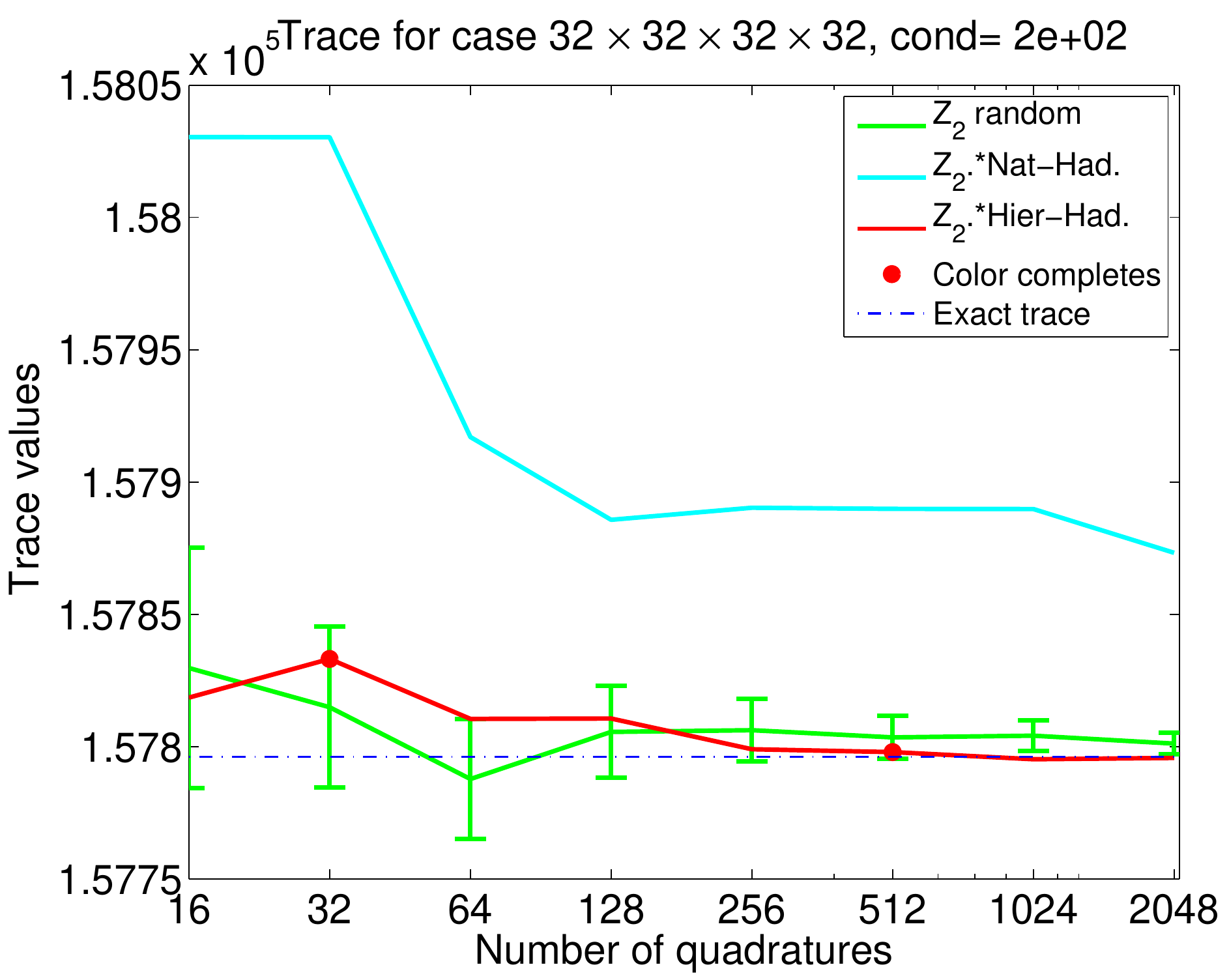}
\caption{ \label{fig:vsZ2smallCond}
Convergence history of $\mathbb{Z}_2$ random estimator, 
  Hadamard vectors in natural order, and hierarchical probing, the latter two 
  with bias removed as in (\ref{eq:removebias}).
Because of small condition number, $A^{-1}$ has a lot of structure, 
  making hierarchical probing clearly superior to the standard estimator. 
As expected, Hadamard vectors in natural order are not competitive.
The markers on the plot of the hierarchical probing method designate
  the number of vectors required for a particular distance coloring 
  to complete. 
It is on these markers that structure is captured and error minimized.
        }
\end{figure}

\begin{figure}[ht]
\includegraphics[width=0.5\textwidth,angle=0]{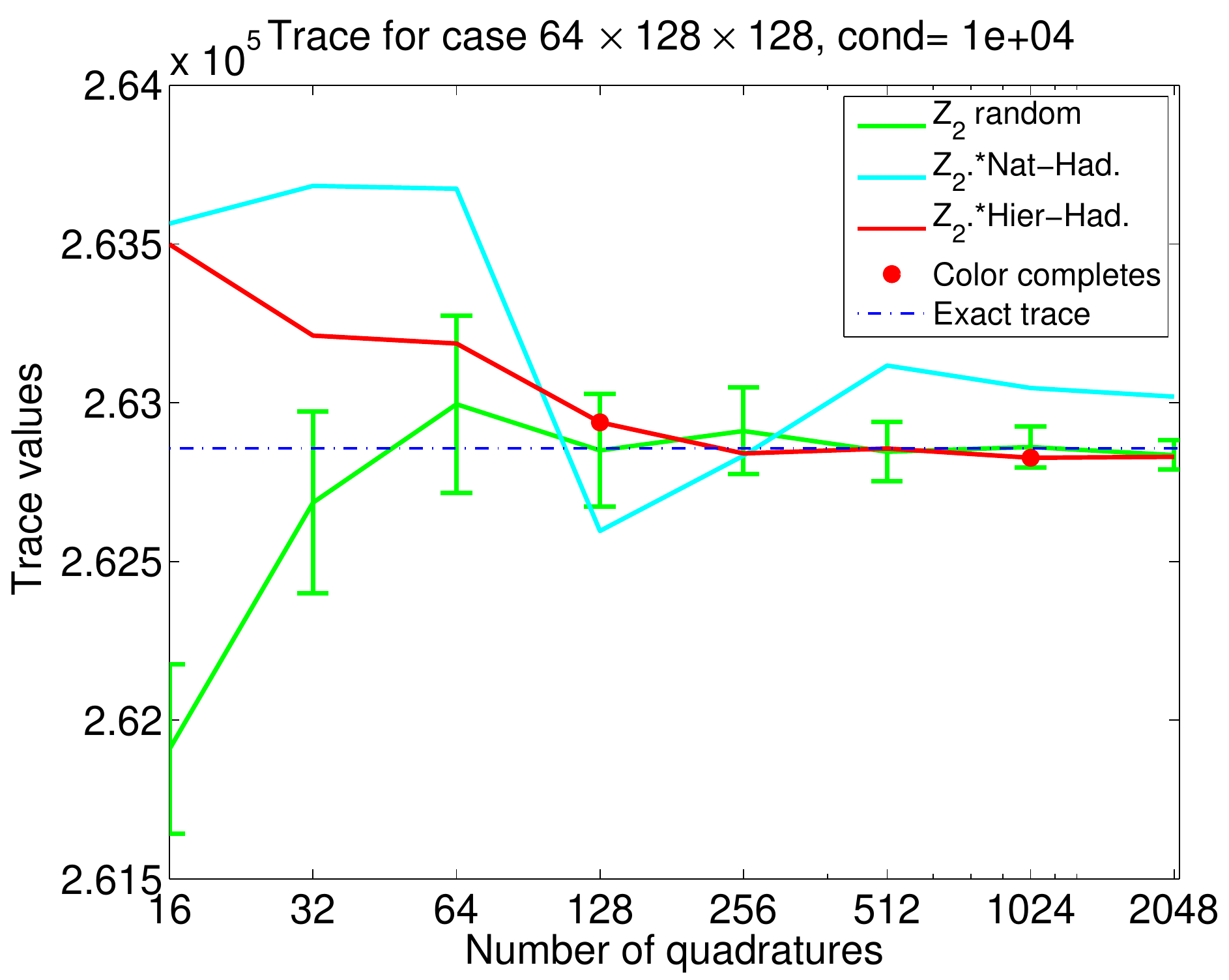}
\includegraphics[width=0.5\textwidth,angle=0]{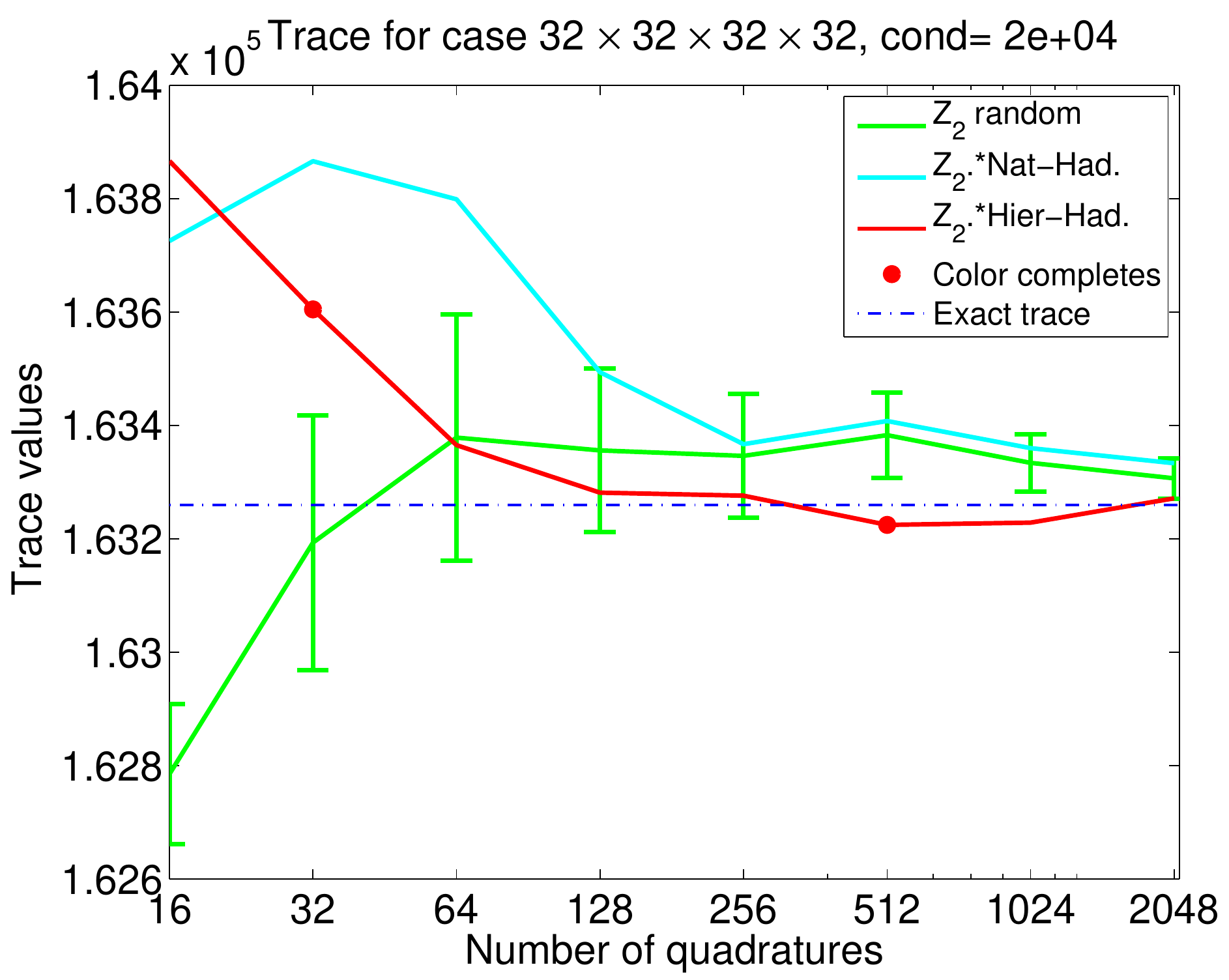}
\caption{ \label{fig:vsZ2midCond}
Convergence history of the three estimators as in Figure \ref{fig:vsZ2smallCond}
  for a larger condition number $O(10^4)$.
As the structure of $A^{-1}$ becomes less prominent, the differences between 
  methods reduce.
Still, hierarchical probing has a clear advantage.
        }
\end{figure}

\begin{figure}[ht]
\includegraphics[width=0.5\textwidth,angle=0]{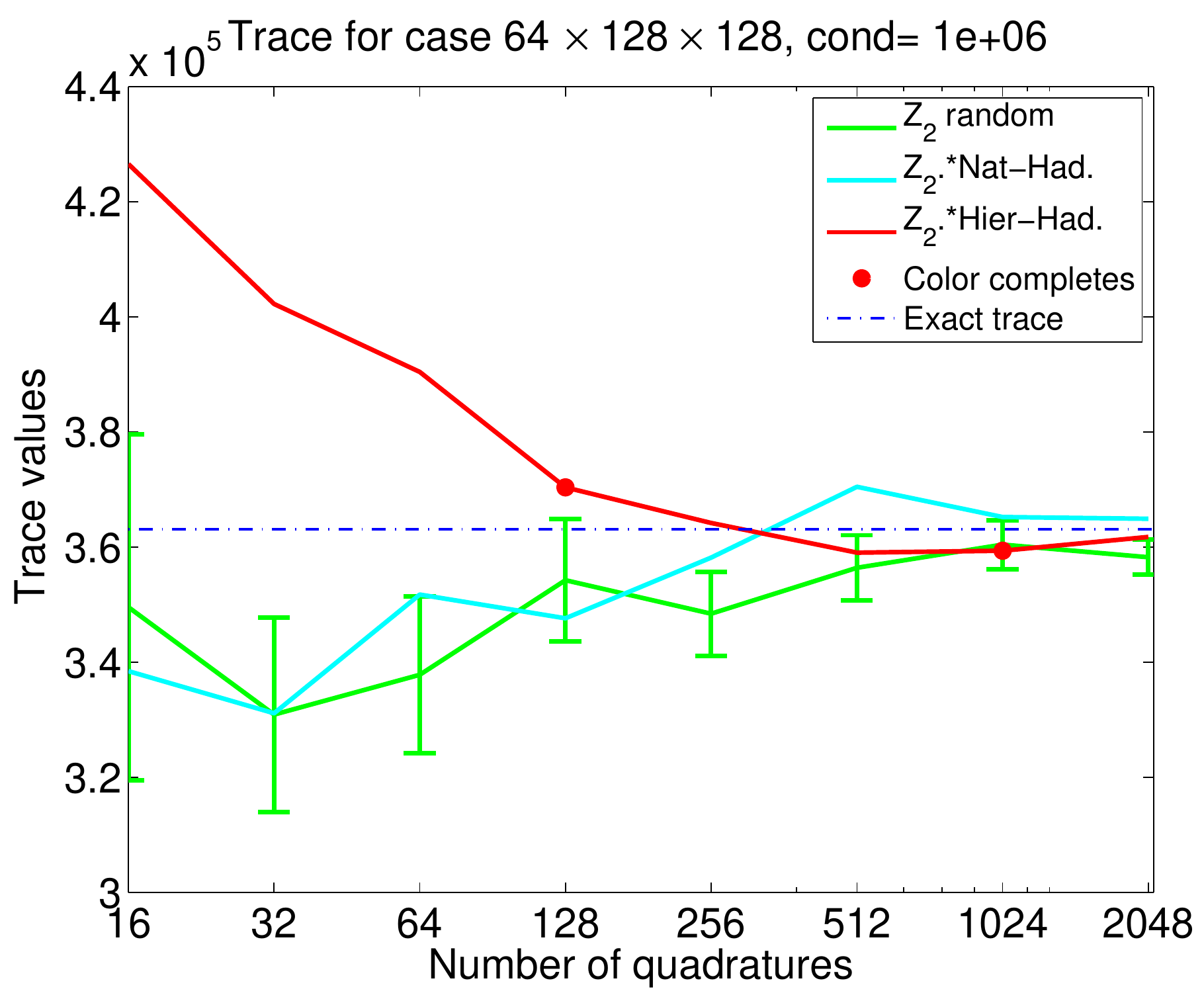}
\includegraphics[width=0.5\textwidth,angle=0]{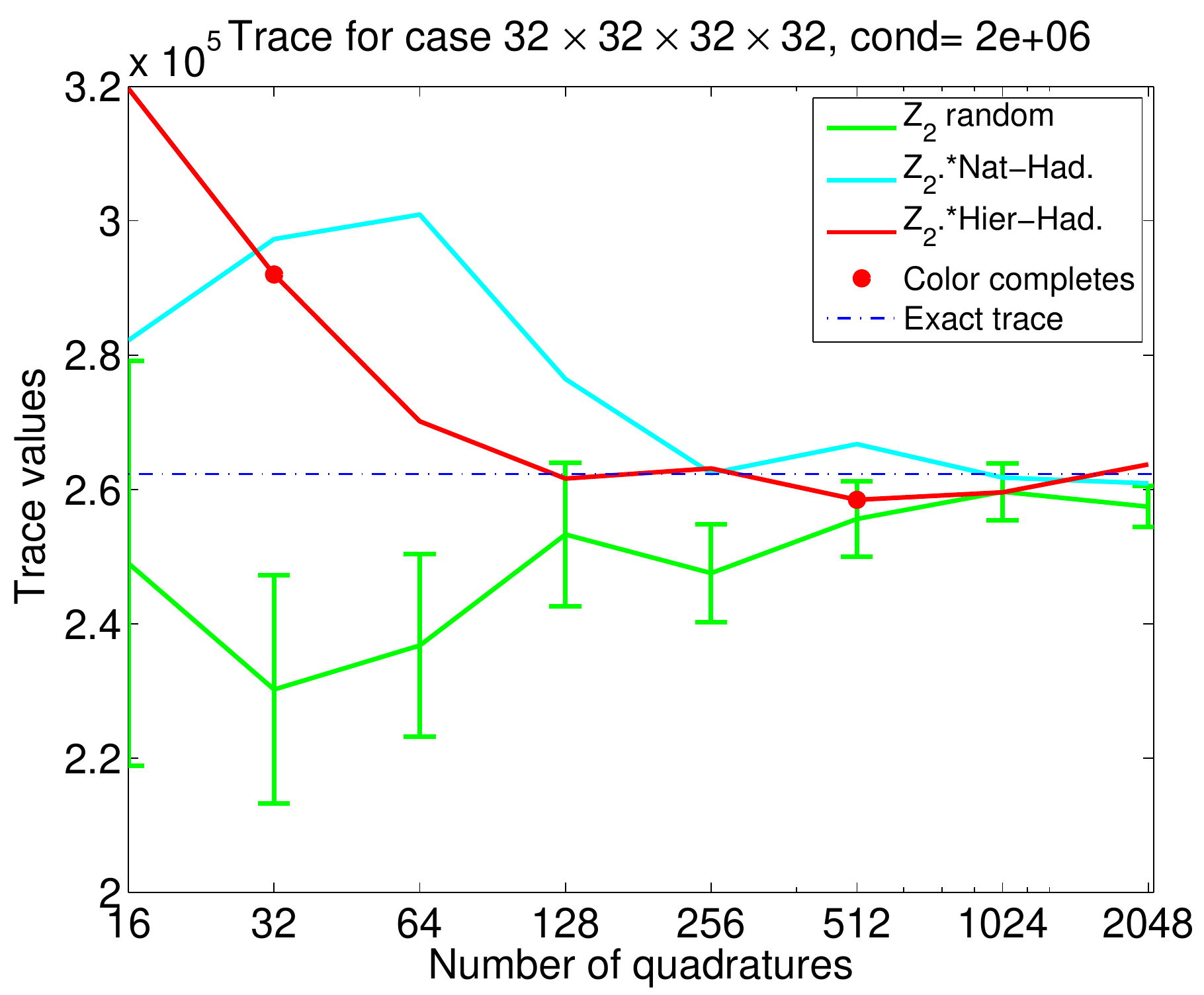}
\caption{ \label{fig:vsZ2largeCond}
Convergence history of the three estimators as in Figure \ref{fig:vsZ2smallCond}
  for a high condition number $O(10^6)$.
Even with no prominent structure in $A^{-1}$ to discover,
  hierarchical probing is as effective as the standard method.
        }
\end{figure}

Once we use a random vector $z_0$ to modify our sequence as
  $z_k\odot z_0$ (\ref{eq:removebias}), hierarchical probing becomes a
  stochastic process, whose statistical properties must be studied.
Thus, we generate $z_0^{(i)}, i=1:100,$ $\mathbb{Z}_2$ random vectors, 
  and for each one we produce a modified sequence of 
  the hierarchical probing vectors.
Then, we use the 100 values $x_k^TA^{-1}x_k$, where 
  $x_k=z_0^{(i)}\odot z_k$, at every step of the 100 MC estimators
  to calculate confidence intervals.
These are shown in Figure \ref{fig:vsZ2_variances}.
We emphasize that the confidence intervals for the $\mathbb{Z}_2$ random estimator 
  are computed differently, based on the $\overline{Var}$ estimator of 
  the preceding MC steps, so they may not be accurate initially.
Even on a 4-D problem, hierarchical probing provides a clear variance 
  improvement.

\begin{figure}[ht]
\centering
\includegraphics[width=0.5\textwidth,angle=0]{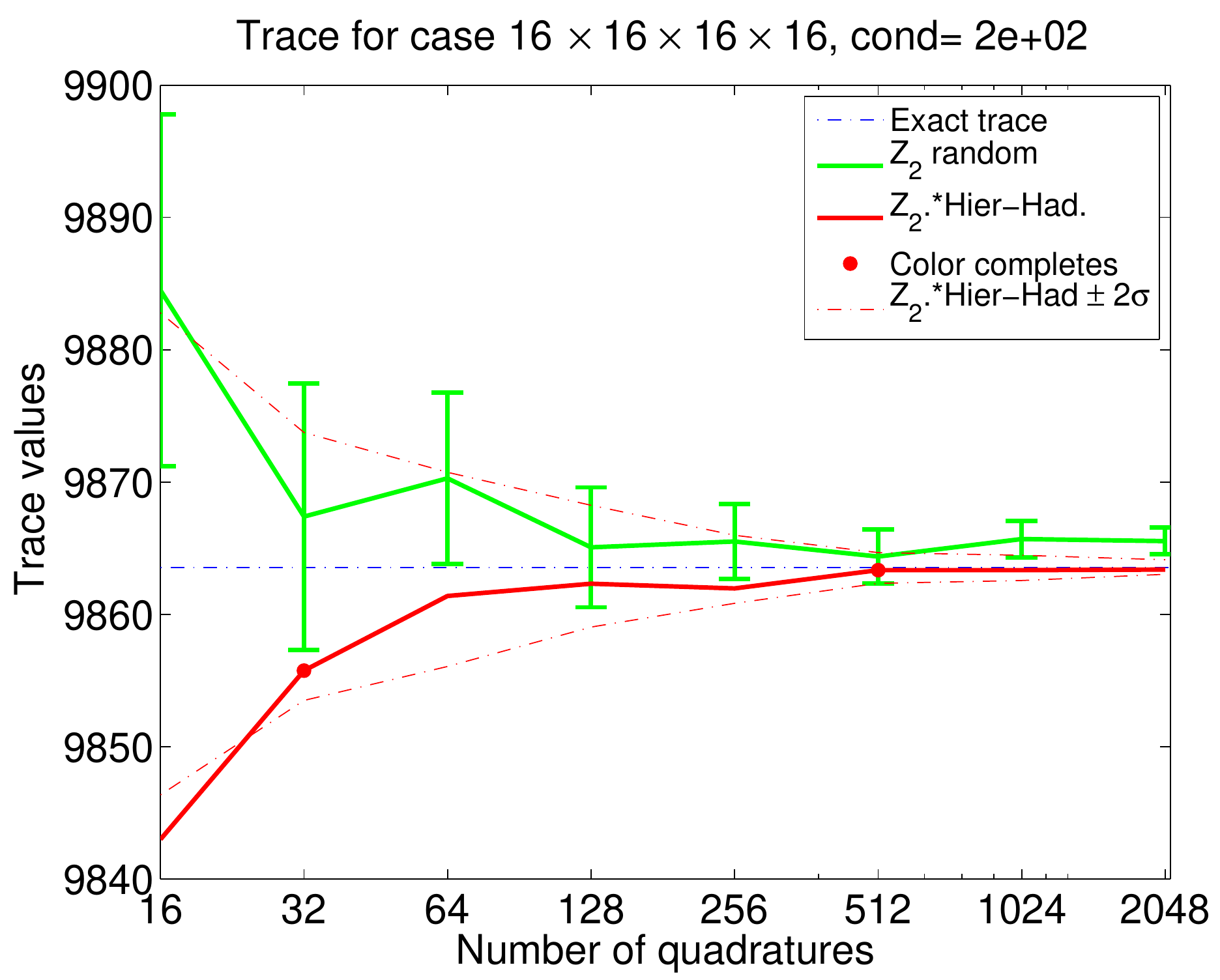}
\caption{ \label{fig:vsZ2_variances}
Providing statistics over 100 random vectors $z_0$, used to modify the 
  sequence of 2048 hierarchical probing vectors as in (\ref{eq:removebias}).
At every step, the variance of quadratures from the 100 different runs 
  is computed, and confidence intervals reported around the hierarchical 
  probing convergence.
Note that for the standard noise MC estimator confidence intervals 
  are computed differently and thus they are not directly comparable.
        }
\end{figure}

\subsection{A large QCD problem}
 The methodology presented in this paper has the potential of improving a multitude of LQCD calculations. In this section, we focus on the
 calculation of $C = \Tr(D^{-1})$, where the Dirac matrix $D$
 is a non-symmetric complex sparse matrix.
  This is representative of a larger class of calculations usually called ``{\em disconnected diagrams}".
 The physical observable $C$ is related to an important property of QCD called spontaneous chiral symmetry breaking.

Our goal is to compare the standard MC approach of computing the trace with our hierarchical probing method.
 Our test was performed on a single gauge field configuration using the Dirac matrix that  corresponds to the  ``{\em strange}" quark Dirac matrix resulting from the Clover-Wilson fermion discretization~\cite{Sheikholeslami:1985ij}. The {\em strange} quark is the third heaviest
quark flavor in nature. The gauge configuration had dimensions of $32^3\times 64$ with a lattice spacing of $a=0.11fm$, 
for a problem size of 24 million. 

First, we used an ensemble of $n=253$ noise vectors to estimate the variance
  of the standard MC method, with complete probing (dilution) of the internal color-spin space of dimension 12 to completely eliminate the 
variance due to connections in this space. Then, for each of these noise vectors, we modified as in (\ref{eq:removebias}) a sequence of hierarchical probing vectors which were generated based on space-time connections. As with the standard MC estimator, full dilution of the color-spin space was performed. This procedure was  performed in order to statistically estimate the
variance of hierarchical probing, similarly to the test in Figure~\ref{fig:vsZ2_variances}.
 In Figure~\ref{fig:qcd_var_speed}(a), we present the variance of the hierarchical probing estimator as a function
of the number of space-time probing vectors in the sequence. The main feature in this plot is that the variance drops as more
vectors are used. Local minima occur at numbers of vectors that are  powers of 2, where all connections of a given Manhattan distance are eliminated from the variance. The uncertainty of the variance, represented by the errorbars in the plot,  is estimated using the Jackknife resampling procedure of our noise vector ensemble. 

In addition to the variance, we estimate the speed-up ratio of the hierarchical probing estimator over the standard MC estimator.
We define speed-up ratio as:
$$
R_s = \frac{V_{stoc}}{V_{hp}(s)\, \times\, s}\,,
$$
where $V_{hp}(s)$ is the variance over the $n$ different runs
  when the $s$-th hierarchical probing vector is used,
and $V_{stoch}$ is the variance of the standard MC estimator as estimated 
  from $n=253$ samples.
The rescaling factor of $s$ is there to account for the fact that if one had been using a pure stochastic noise with $n\times s$ vectors,
the variance would be smaller by a factor of $s$.  Thus, the variance comparison is performed on equal amount of computation for both methods. In Figure~\ref{fig:qcd_var_speed}(b) we present the speed-up ratio $R_s$ as a function $s$. The errorbars on $R_s$ are estimated using Jackknife resampling from our ensemble of starting noise vectors. The peaks in this plot occur at the points where $s$ is a power of 2, as in the variance case. A maximum overall speed-up factor of about 10 is observed at $s=512$. Note that the color completion points for this experiment are at $s=2$, $s=32$ and $s=512$ vectors.

\begin{figure}[ht]
\includegraphics[width=0.5\textwidth]{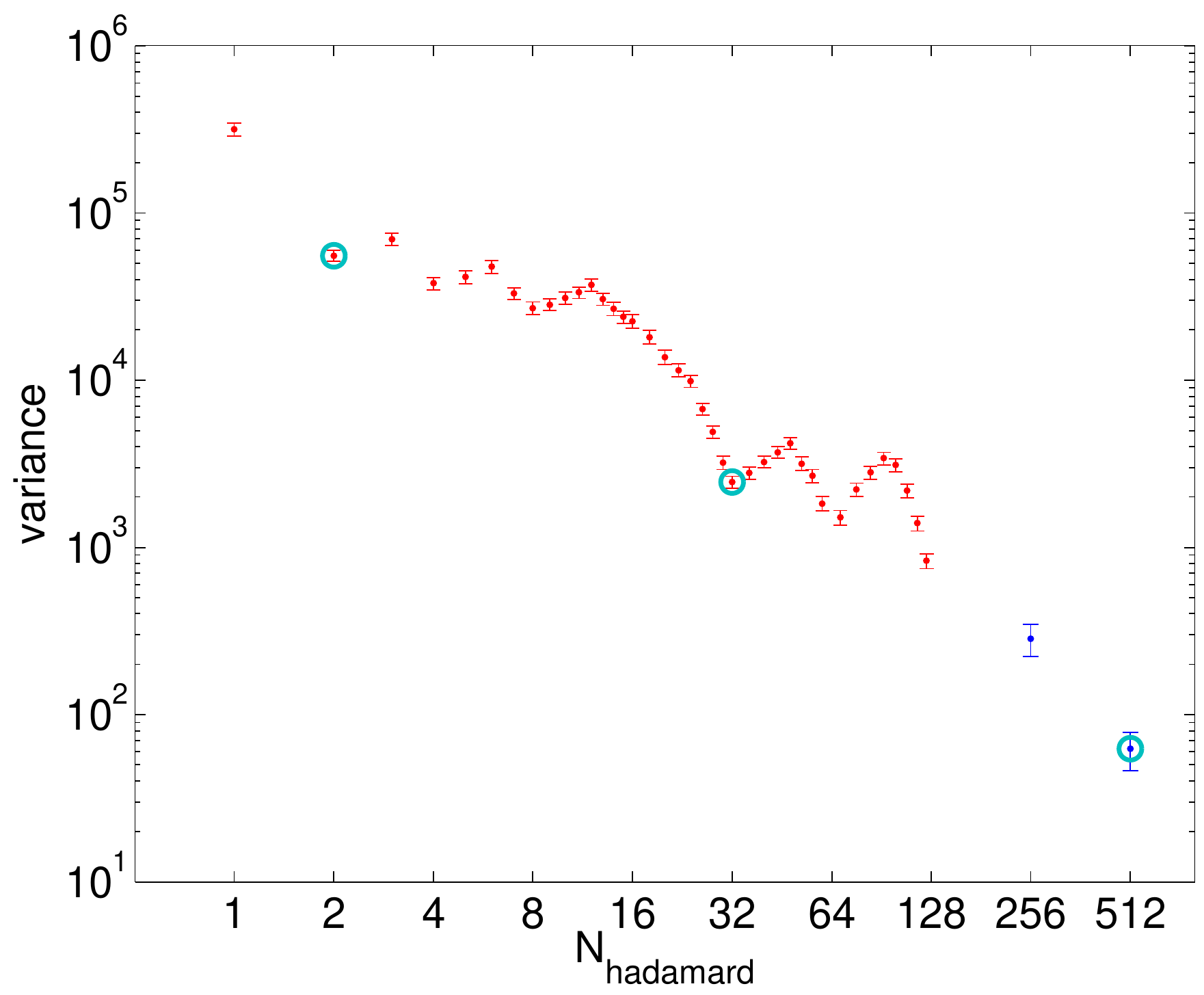}
\includegraphics[width=0.5\textwidth]{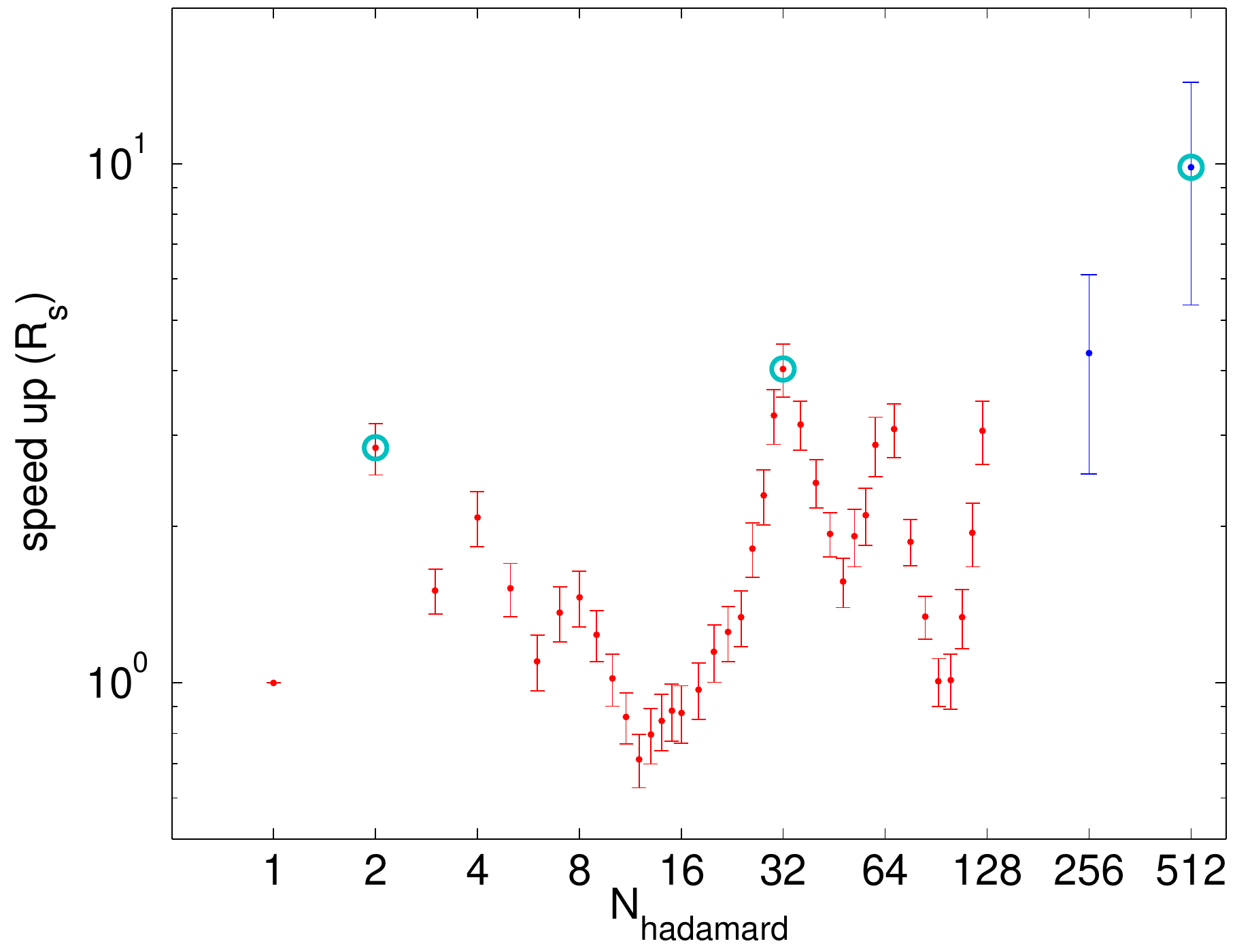}
\caption{\label{fig:qcd_var_speed} 
(a) Left: The variance of the hierarchical probing trace estimator as 
  a function of the number of vectors ($s$) used. The minima appear 
  when $s$ is a power of two. 
 The places where the colors complete are marked with the cyan circle. 
 These minima become successively deeper as we progress from 2 to 32 to 512 vectors.  
(b) Right: Speed-up of the LQCD trace calculation over the standard 
  $\mathbb{Z}_2$ MC estimator. The cyan circles mark where colors complete.
  The maximal speed up is observed at $s=512$.
In both cases the uncertainties are estimated using the Jackknife procedure 
  on a sample of 253 noise vectors, except for $s=256$ and $512$ 
  where 37 vectors were used.
}
\end{figure}

\colorred{
Finally, we report on a comparison with classic probing 
  for this large QCD problem.
There is a variety of approaches for efficient distance-2 coloring 
  in the literature \cite{Pothen_Coloring,Parallel_dist_2_color}, 
  but we have not found any extensions to distance-$k$ coloring.
To provide a realistic comparison, we have implemented two different 
  distance-$k$ coloring algorithms.
}

\colorred{
The first is a more general method that applies Dijkstra's algorithm 
  on each node to produce the lengths of the shortest paths from this 
  node to all nodes up to a selected maximum distance $k_m$.
We use a Fibonacci heap implementation to exploit sparsity.
With this array available, we can then distance-$k$ color this node 
  for one or more $k \leq k_m$, at the same time.
A parallel implementation of this naturally sequential algorithm
  is rather involved (see \cite{Parallel_dist_2_color} for the distance-2
  version), and it may result in more colors.
Instead, we run it for a maximum distance $k_m=16$ on a state-of-the-art, 
  Intel Xeon X5672, 3.2GHz server.
After several days of runtime, we extrapolated that total execution time
  would be one year!
}

\colorred{
The second method is based on the fact that the distance-$k$ neighborhood 
  of a lattice node is explicitly known geometrically.
We implemented a coloring algorithm that visits only this neighborhood
  for each node, thus achieving the minimum possible complexity
  for this problem \cite{Parallel_dist_2_color}.
The distance-$4$ coloring of our LQCD lattice produced 123 colors and took
  457 seconds 
  on the above Xeon server.
Using four random vectors for each of these probing vectors (so that 
  the total number of quadratures is similar to our hierarchical probing),
  we ran 50 sets of experiments, and measured the variance of classical probing.
We found that its variance was 2.16 times larger than our hierarchical probing,
  or in other words, our method was 2.16 times faster.
This is expected as we explained earlier.
Finally, note that computing the quadratures took 4 hours on four GPUs, 
  on a dedicated machine for LQCD calculations.
Even though classic probing with distance-$4$ is feasible for this problem, 
 computing the distance-8 coloring requires 5377 seconds, which
 becomes comparable to the time for computing the quadratures. 
Contrast that to the 2 seconds needed to compute the hierarchical probing.
}

\section{Conclusions}
The motivation for this work comes from our need to compute
  $\TrAi$ for very large sparse matrices and LQCD lattices.
Current methods are based on Monte Carlo and do not sufficiently exploit 
  the structure of the matrix.
Probing is an attractive technique but it is very expensive and cannot 
  be used incrementally. 
Our research has addressed these issues.

We have developed the idea of hierarchical probing that produces
  suboptimal but nested distance-$k$ colorings recursively, for all 
  distances $k\geq 1$ up to the diameter of the graph.
We have adapted this idea to uniform lattices of any dimension 
  in a very efficient and parallelizable way.

To generate probing vectors that follow the hierarchical permutation,
  and can be used incrementally to improve accuracy,
  we have developed an algorithm that produces a specific 
  permutation of the Hadamard vectors.
This algorithm is limited to cases where the number of colors produced 
  at every level is a power of two.
We have also provided a recursive algorithm based on Fourier matrices
  that provides the appropriate sequence under the weaker assumption
  of having the same number of colors per block within a single level.
These conditions are satisfied on toroidal lattices.
Finally, we proposed an inexpensive technique to avoid deterministic 
  bias while using the above sequences of vectors.

We have performed a set of experiments in the context of computing $\TrAi$,
  and have shown that providing a hierarchical coloring for all 
  possible distances is to be preferred over classic probing 
  for a specific distance.
We also showed that our methods provide significant speed-ups over 
  the standard Monte Carlo approach.

We are currently working on combining the ideas in this paper 
  with other variance reduction techniques, in particular 
  deflation type methods.
Hierarchical coloring ideas might also be useful for general sparse
  matrices for trace computations or in preconditioning.

\section*{Acknowledgements}
Support for this work has been provided by NSF under a grant No. CCF-1218349, 
  through the Scientific Discovery through Advanced Computing (SciDAC) program funded by U.S. Department of Energy, Office of Science, Advanced Scientific Computing Research and Nuclear Physics under award number DE-FC02-12ER41890, and the Jeffress memorial grant.
Some of the experiments were performed using computational facilities at the College of William and Mary which were provided with the assistance of the National Science Foundation, the Virginia Port Authority, Sun Microsystems, and Virginia's Commonwealth Technology Research Fund.

\bibliography{}
\bibliographystyle{SIAM/siam}
\end{document}